\documentclass[leqno,12pt]{article}
\usepackage[T1]{fontenc}
\usepackage{amsmath}
\usepackage{float}
\usepackage{amsfonts}
\usepackage{stackengine}
\usepackage{graphicx}
\usepackage[mathscr]{euscript}
\usepackage{rotating}
\usepackage{pdflscape}
\usepackage[hmargin=2.5cm,vmargin=3cm]{geometry} 
\usepackage{booktabs,caption}
\usepackage{natbib}
\usepackage{setspace}
\usepackage{bbm}
\usepackage{amsthm}
\usepackage{enumerate}
\usepackage{threeparttable}
\usepackage{verbatim}
\usepackage{caption}
\usepackage{url}
\usepackage{amssymb}
\usepackage{subcaption,array}
\usepackage{adjustbox}
\usepackage[titletoc,title]{appendix}
\usepackage{multirow}
\usepackage{threeparttable}

\usepackage[hidelinks,breaklinks]{hyperref}
\hypersetup{colorlinks=true,allcolors=blue}

\usepackage{framed}

\newtheorem{theo}{Theorem}

\newtheorem{lemma}{Lemma}
\newtheorem*{remark}{Remark}

\theoremstyle{definition}

\usepackage{calligra}
\usepackage{diagbox}

\usepackage{tikz}
\usetikzlibrary{positioning,shapes.misc}
\usetikzlibrary{shapes.geometric,shapes,arrows.meta,positioning,decorations.pathreplacing}

\tikzset{
    myboxcircle/.style={circle,draw=black,align=center},
}
\tikzset{
    myboxrounded/.style={rounded rectangle,draw=black,align=center},
}
\tikzset{
    myboxrectangle/.style={rectangle,draw=black,align=center},
}

\begin{document}

\title{
\textbf{Demand Analysis under Price Rigidity and Endogenous Assortment: An Application to China's Tobacco Industry}\thanks{We thank Victor Aguirregabiria, Loren Brandt, JoonHwan Cho, Xavier D’Haultfoeuille and seminar participants for their helpful comments. We owe special thanks to Geoffrey T. Fong, Anne C.K. Quah, Steve Xu, Gang Meng, and Ce Shang for providing access to International Tobacco Control survey data and insight into the tobacco industry. All errors that remain are our own.}
}

\author{
    Hui Liu
    \footnote{Ma Yinchu School of Economics, Tianjin University, Tianjin, China, \href{mailto: liuhui33@tju.edu.cn}{liuhui33@tju.edu.cn}.} 
    \\ \emph{Tianjin University}
    \and
    Yao Luo\footnote{Department of Economics, University of Toronto. 150 St. George Street, Toronto, ON, M5S 3G7, Canada, \href{mailto: yao.luo@utoronto.ca}{yao.luo@utoronto.ca}.} 
    \\ \emph{University of Toronto}
    }

\date{January 27, 2025 
}

\maketitle

\thispagestyle{empty} 

\begin{abstract}

We observe nominal price rigidity in tobacco markets across China. The monopolistic seller responds by adjusting product assortments, which remain unobserved by the analyst. We develop and estimate a logit demand model that incorporates assortment discrimination and nominal price rigidity. We find that consumers are significantly more responsive to price changes than conventional models predict. Simulated tax increases reveal that neglecting the role of endogenous assortments results in underestimations of the decline in higher-tier product sales, incorrect directional predictions of lower-tier product sales, and overestimation of tax revenue by more than 50\%. Finally, we extend our methodology to settings with competition and random coefficient models. 

\vspace{0.4cm}

\vspace{0.4cm}
\end{abstract}


\newpage
\setcounter{page}{1}

\setstretch{1.5}


\section{Introduction}

Demand estimation plays a critical role in guiding firms' optimal decision-making and allows researchers to evaluate the impact of public policy on social welfare across different market settings. Many studies implicitly assume that all consumers in a market face the same set of available choices. While this assumption simplifies the analysis, an expanding body of research (\cite{hickman2016demand}, \cite{crawford2021survey}) underscores its limitations in accurately reflecting retail environments shaped by factors like assortment variations, stockouts, and consumers’ limited consideration sets. Overlooking heterogeneity in choice sets can result in biased demand estimates and flawed conclusions regarding public policy outcomes. Addressing this issue is often complicated by the unavailability of detailed data on individual choice sets, making it challenging to account for unobserved choice set variation.


We propose a novel approach to demand estimation that accounts for unobserved choice set heterogeneity while relying solely on aggregate sales data. Unlike prior studies, our method infers unobserved choice sets for different consumers by utilizing the optimal assortments derived from the firm's profit maximization conditions. Many existing studies, such as \citet{crawford2016demand}, \citet{aguiar2021identification}, \citet{chib2023scalable}, and \citet{kashaev2024peer}, rely on observed individual choices, rendering their methods unsuitable for aggregate sales data, which is the focus of our research. We construct and estimate a model in which the seller manages product assortments but cannot adjust prices due to external factors like government regulations, resulting in price rigidity. Price rigidity is common in heavily regulated industries such as tobacco, alcohol, and pharmaceuticals, as documented by \citet{dubois2020effect} in Norway and \citet{iizuka2012physician} in Japan. It also frequently manifests as uniform pricing, as highlighted by studies such as \citet{dellavigna2019uniform}, \citet{adams2019zone}, \citet{NBERw26306}, and \citet{ozhegova2023assortment}. In our framework, the firm instead determines product assortments endogenously to maximize expected profits.

On the theoretical side, we demonstrate that a monopoly firm can streamline its assortment choices using a "foldable menu" approach. Specifically, we show that the firm can simplify its decision-making by selecting from a predetermined set of assortments: \{1\}, \{1,2\},..., up to \{1,2,...,J\}, when products are ranked in descending order of unit margin. For high-income consumers, offering only product 1 is the most profitable strategy, while for low-income consumers, a broader range of options should be made available.\footnote{We refer to this as the "foldable menu" model because it mirrors a seller tailoring a product menu to different consumer types: metaphorically "folding" away all but the most profitable option for high-income consumers and "unfolding" a wider array of choices for lower-income consumers.} This finding is particularly valuable, as the exhaustive search for optimal assortments becomes computationally burdensome with a large product set, given the exponential growth in the number of possible assortments as the product range increases.\footnote{For instance, a firm with 20 products would face $2^{20}$ (i.e., over 1 million) possible assortments.} Estimating the foldable menu model requires just one additional step beyond the standard demand framework: identifying the optimal assortment for each consumer. Finally, we expand our methodology to include settings with competition and random coefficients. We establish that the foldable menu result remains valid in these contexts, proving the existence of at least one equilibrium in assortment competition. This equilibrium can be efficiently identified through an iterative process. These findings emphasize broad relevance of our approach to various empirical applications.

On the empirical side, we apply the proposed model and estimation method to data from China’s tobacco industry. In line with standard industry practice, we categorize cigarettes into five tiers (Tiers I to V) based on price, from highest to lowest. Our empirical results show that consumers are significantly more responsive to price changes under the foldable menu model compared to models that ignore choice set heterogeneity. Additionally, endogenous assortment plays a key role in driving the sales of high-priced cigarettes in wealthier provinces, where low-priced options are less available compared to low-income provinces. Ignoring endogenous assortment would lead to an overestimation of the income effect, resulting in an understatement of consumers’ price sensitivity, as income and price sensitivity are negatively correlated. This application highlights how overlooking endogenous assortment can lead to biased demand estimates and, in turn, misinformed policy conclusions.

To further illustrate the impact of ignoring endogenous assortments, we calculate the percentage change in sales when the price of one tier increases by 1\%, holding other prices constant. We find that own-price elasticities are significantly higher than those suggested by the standard model. For instance, only Tier V is inelastic under our model, whereas the standard model would indicate that Tiers III through V are all inelastic. Moreover, lower tiers exhibit larger cross-price elasticities when the price of a higher tier increases. For example, a 1\% price increase in Tier I has the greatest effect on Tiers IV and V. In contrast, the standard model would predict a more uniform effect across competing tiers. We also calculate the percentage change in sales in response to a 1\% price increase across all tiers. While the sales of the top tiers decline, Tiers IV and V experience an increase due to the firm's endogenous adjustment of assortments, which increases the availability of lower-tier products. By contrast, the standard model would predict a decrease in sales across all tiers.

We use the estimated model to assess the impact of a tax increase (ranging from 5\% to 20\%) on cigarette sales and tax revenue. Under the foldable menu model, we observe a substantial decline in sales for Tiers I to III, while sales for Tiers IV and V increase. In contrast, the standard model predicts a decrease in sales across all tiers. Since higher tiers carry significantly higher unit taxes, the standard model overestimates the effect of the tax increase on total tax revenue by 58\% to 74\%. Additionally, we quantify the changes in sales for each cigarette tier when all tiers are made available. The results indicate that, when all tiers are accessible, sales in Tiers I to III decline, while Tiers IV and V see a significant boost. Overall, total cigarette sales are projected to increase by 37.81\%, though wholesale profits are expected to decline by approximately 13.71\% in this scenario.

\section*{Related Literature}

Many existing studies rely on the observation of individual choices and panel data for demand estimation with unobserved choice sets. For instance, \cite{chiang1999markov} proposes an Bayesian approach based on an MCMC sampling procedure of unobserved consideration set and applies it to scanner panel data. See also \citet{chib2023scalable}. \cite{goeree2008limited} simulates each consumer's unobserved consideration set based on their model-derived awareness of products through advertising. In recent developments, \cite{crawford2016demand} constructs "sufficient sets" of products to "difference out" unobserved consideration sets by exploiting information in repeated purchase decisions.  \cite{aguiar2021identification} utilizes panel data to jointly identify unobserved choice sets and preferences. \cite{lu2021estimating} develops a moment inequality estimator for partially identify unobserved choice sets based on the bounds on each consumer's actual choice set and the monotonicity of the choice function with respect to the choice set. \cite{kashaev2024peer} incorporates peer effects to model the formation of agents' unobserved consideration sets.\footnote{\cite{tenn2008biases} and \cite{bruno2008research} rely on market-level "All Commodity Value" (ACV) data to simulate unobserved choice sets. See also the theoretical literature about limited consumer attention, such as \cite{eliaz2011consideration}, \cite{masatlioglu2012revealed}, and \cite{manzini2014stochastic}.}

Our paper also relates to a body of research examining heterogeneous choice sets driven by firms' strategic assortment decisions. Like our work, studies such as \cite{draganska2009beyond}, \cite{iaria2014consumer}, \cite{musalem2015demand}, \cite{shah2015diagnosing}, and \cite{aguirregabiria2023identification} endogenize assortments in their models. However, these studies assume identical assortments within a market and primarily focus on product entry, whereas we allow for heterogeneous assortments within the same market. Additionally, their empirical analyses require observable assortment variation, while our approach does not. Another related study is \cite{dubois2020effect}, which also assumes price rigidity and arrives at a similar conclusion: in a competitive setting with random coefficients, each pharmacy chain selects assortments from two drugs, and no chain exclusively offers the drug with the lower margin. Our Lemma \ref{lemma 6} generalizes this result to accommodate any number of products. More recently, \cite{ozhegova2023assortment} illustrates how retailers strategically adjust product assortments to respond to local market conditions, such as average willingness to pay and competition, when prices are set at the national level. This study uses store-level sales data and treats the products with positive sales as the available assortment. In contrast, our paper utilizes sales data aggregated at the provincial level and endogenizes assortments across different areas within each market.

In addition to limited consideration and assortment variation, choice set heterogeneity can also arise from stockouts. Since stockouts occur at the point of sale, studies addressing demand estimation in the presence of stockouts typically rely on store-level sales and inventory data, whereas our study relies solely on aggregate sales data. For demand estimation with stockouts, see, for example, \cite{anupindi1998estimation}, \cite{aguirregabiria1999dynamics}, \cite{musalem2010structural}, and \cite{conlon2013demand}. Our paper also relates to recent research on uniform pricing, which describes the practice of charging the same prices across markets with varying demographics, preferences, and competition levels. For further discussion, see \cite{dellavigna2019uniform}, \cite{adams2019zone}, and \cite{NBERw26306}.

Our paper contributes to a growing body of literature that utilizes optimization conditions to recover unobserved model primitives. For instance, \cite{luo2018structural} uncovers the distribution of consumer types by leveraging the first-order conditions of both consumers and the firm, while \cite{d2019automobile} retrieves unobserved auto sales prices through dealers' first-order conditions. Similar approaches can be found in \citet{luo2023bundling} and \citet{aryal2024identification}. Additionally, \citet{dubois2018identifying} employ firms' optimality conditions to infer price ceilings for drug prices, which are unknown to the analyst. In our study, we apply the firm's profit maximization conditions to recover unobserved assortments.

Lastly, our paper relates to several studies on China's tobacco industry. For example, \cite{gao2012can} examines China's cigarette pricing mechanism, while \cite{goodchild2018early} and \cite{zheng2018tobacco} evaluate the impact of tax increases on the industry. Additionally, \cite{lance2004cigarette}, \cite{bishop2007chinese}, \cite{liu2015smokers}, and \cite{li2016heterogeneous} estimate demand elasticity in China's tobacco market.

The remainder of the paper is organized as follows: Section 2 provides a brief overview of the institutional background of China's tobacco industry. Section 3 introduces our theoretical model and outlines its key properties. Section 4 summarizes the data from China's tobacco industry and discusses the functional form used in the application. Section 5 applies the model to estimate cigarette demand in China and conducts post-estimation analyses. Section 6 discusses extensions to incorporate competition and random coefficients. Section 7 concludes.

\section{Institutional Background} 

The China National Tobacco Company (CNTC) and the State Tobacco Monopoly Administration (STMA) were established in 1983, with the mandate to oversee tobacco leaf production, cigarette manufacturing, and nationwide marketing in China. STMA, as a government agency, sets broad policies and delegates full authority to CNTC for decisions related to tobacco production quotas, leaf procurement, transportation, storage, and the manufacturing and sales of cigarettes \citep{hu2008tobacco}.

CNTC operates as a state-owned enterprise with a monopoly over the cigarette market, accounting for 98\% of domestic sales. For instance, in 2015, CNTC reported a gross profit of 303 billion RMB (approximately 48 billion USD), making it the most profitable company in the country. The tobacco industry, led by CNTC, also made a significant contribution to the national economy, generating 840.4 billion RMB (about 122 billion USD) in tax revenue from tobacco in 2015. Overall, CNTC contributes 7\%–10\% of the Chinese central government's annual revenue through tobacco taxes and profit-sharing, not including revenues shared with local governments \citep{xu2019impact}.

China uses a five-tier pricing structure for cigarettes, determined primarily by the quality of the tobacco leaves used. STMA formulates cigarette "allocation plans," setting the allocation price in the process. This price is what tobacco producers charge wholesalers. Since June 2001, the allocation price has replaced the producer price as the tax base for the cigarette ad valorem excise tax. The tax authority uses the allocation price for tax collection, while STMA employs it to categorize cigarettes into five tiers (Tier I to Tier V), with Tier I representing the highest-priced cigarettes \citep{gao2012can}.

STMA regulates cigarette retail prices by directly setting and controlling the retail prices of all cigarette brands through its licensing system, as only licensed retailers are legally permitted to sell cigarettes in China \citep{zheng2018tobacco}. By using the allocation price as a baseline, STMA can effectively control both nominal wholesale and retail prices by adjusting the allocation-wholesale margin or the wholesale-retail margin \citep{gao2012can}. Cigarettes are uniformly priced within each provincial market using the "Retail Cigarette Price Tag." Price adjustments are applied consistently across the entire province.\footnote{See \hyperlink{http://www.etmoc.com/Firms/Prodlist?Id=44170}{TobaccoMarket (2021)} for an article on price adjustments.} The printing and distribution of these tags are managed by municipal-level tobacco monopoly authorities, with the method and content of price displays overseen by municipal price regulatory agencies. Unauthorized printing or distribution of these tags is strictly prohibited. Vendors who manipulate prices may face severe penalties, including license suspension. Additionally, price regulatory departments are authorized to impose sanctions in accordance with the relevant provisions of the "People's Republic of China Price Law."\footnote{See \hyperlink{http://www.etmoc.com/look/Lawlist?Id=1804}{TobaccoMarket (2007)} for an article on price regulations.}

Despite the availability of all cigarette tiers at the provincial level, there is widespread acknowledgment of the scarcity of low-priced cigarettes (priced under ¥10) in many local markets.\footnote{See \hyperlink{http://www.etmoc.com/market/Newslist?Id=20422}{TobaccoMarket (2009)} for an article on scarcity of low-priced cigarettes.} This scarcity coincides with CNTC's Premiumization Strategy, introduced in 2009, which aims to encourage smokers to upgrade to more expensive brands that significantly contribute to tax revenue and industry profits \citep{xu2019impact}. While demand for low-priced cigarettes persists, limiting their availability is an effective approach to profit maximization under price rigidity. In rural areas, where lower-income smokers predominate, the focus remains on providing lower-tier cigarettes to boost overall sales, while in urban markets, the emphasis is on selling higher-tier products to enhance profitability. This distribution strategy explains the disparity between the availability of low-priced cigarettes at the provincial level and their scarcity in local communities. 

Industry insiders from the International Tobacco Control confirm these observations, noting that low-priced cigarettes are intentionally made more accessible primarily in low-income communities. This aligns with the government's supervision of China's tobacco industry, where the CNTC, which oversees distribution and sales, holds the authority to determine the assortment of cigarettes supplied to licensed retailers. The strategy of limiting low-priced cigarettes to low-income communities is consistent with the theoretical predictions of our model, where CNTC endogenously selects assortments of cigarettes to maximize profits.

\section{Theoretical Model}

\subsection{The Unobserved Assortment Discrimination Model} 

Building on the previously discussed institutional context, we develop a logit model that incorporates both price rigidity and endogenous product assortment. First, firms are required to maintain uniform pricing within each market. Given the strict regulation of cigarette prices in China, the assumption of nominal price rigidity is likely to hold, though we allow for price variation across different markets. Second, firms have access to local demographic information, enabling them to make endogenous decisions about product assortments that maximize expected profits from each consumer. Consequently, our model features product assortments that vary within a market, which remain unobserved by researchers. This contrasts with traditional demand models, where the same set of products is assumed to be universally available to all consumers within a market. 

We assume that consumers select products exclusively from locally available assortments without searching across locations---a behavior well-documented in China’s tobacco industry. With online cigarette sales prohibited in China, consumers are restricted to purchasing from local retailers, such as convenience stores, gas stations, and smoke shops. The combination of relatively low and rigid cigarette prices, along with the costs associated with searching and transportation, further discourages consumers from seeking products outside their immediate vicinity. Additionally, the endogenous assortment of low-priced cigarettes, which are more common in low-income and rural areas, limits their availability for urban consumers. The absence of such products locally often extends to nearby communities as well. These factors collectively support the assumption that Chinese smokers predominantly choose from locally available assortments, affirming the relevance and applicability of the foldable menu model in this context.

Hereafter, a "market" refers to a province-year, a "consumer" represents a typical (potential) smoker in a local community, and a "product" corresponds to a cigarette tier. Each "firm" is one of the provincial CNTC branches responsible for distributing various assortments of cigarettes across areas within the respective provincial market. For clarity, we will focus on a single market and temporarily omit market-specific subscripts.

A single firm serves the market by offering a range of product assortments, drawn from a pool of $J$ distinct products, to $I$ consumers. While the firm does not control prices, it has the flexibility to assign different assortments to each consumer. As a result, each consumer's choices are limited to the products included in their assigned assortment. The utility that consumer $i$ derives from selecting product $j$ follows the standard form:
\begin{align}
u_{ij} = \delta_{ij} + \epsilon_{ij}, \; \forall j\in \mathcal{J}_i \cup \{0\}, 
\end{align}
where $\delta_{ij}$ denotes the deterministic component of utility, $\epsilon_{ij}$ represents the idiosyncratic error term, and $\mathcal{J}_i$ is the product assortment available to consumer $i$. We assume that the firm can estimate $\delta_{ij}$ based on observable product characteristics and consumer demographics. However, the firm only has knowledge of the distribution of the idiosyncratic error terms $\epsilon_{ij}$, which, as is standard in the literature, follows an extreme value type-I distribution. Alternative $0$, i.e. the outside alternative, represents the option to not purchase any of the available products, and we normalize the utility to $u_{i0} = \epsilon_{i0}$. 

Consumer $i$ selects product $j$ if and only if
\begin{equation*}
u_{ij} \geq u_{ij'}, \; \forall j' \neq j,
\end{equation*}


The firm's objective is to provide consumer $i$ with the assortment that maximizes expected profit, denoted as the optimal assortment $\mathcal{J}^*_i$. Let $\vec{J} = \{1,2,...,J\}$ be the complete assortment and $\pi_j$ be the per unit profit margin of product $j$. For any subset $\mathcal{J} \subseteq \vec{J}$, the expected profit from consumer $i$ is:
\begin{equation}
E\Pi_{i,\mathcal{J}} = \sum_{j \in \mathcal{J}} \underbrace{\pi_j}_{\text{profit margin}} \underbrace{\frac{\exp(\delta_{ij})}{1 + \sum_{j \in \mathcal{J}}\exp(\delta_{ij})}}_\text{choice probability}.
\end{equation}
where the choice probability follows the standard logit properties. 

Searching for the optimal assortment $\mathcal{J}^*_i$ among all possible assortments can be computationally cumbersome, due to the aforementioned exponential nature of assortments. However, the following section demonstrates that the firm only needs to choose from a foldable menu, wherein the number of assortments equals the number of products.

\subsection{Foldable Menu}

We first prove that the optimal assortments can only occur in the foldable menu containing complete sets $\{\vec{1},\vec{2},\cdots,\vec{J}\}$.\footnote{A similar result is shown in \cite{talluri2004revenue} for revenue management. Here, we prove the result from a profit maximization perspective.} Without loss of generality, we can arrange the products in descending order as $\pi_1 \geq \pi_2 \geq ... \geq \pi_J \geq \pi_0 = 0$.

\begin{lemma}
    An optimal assortment must contain product 1.
\end{lemma}

\begin{proof}
    See Appendix \ref{Lemma 1}.
\end{proof}

\begin{lemma}
   If $E \Pi_{i,\mathcal{J}\cup\{m\}} \leq E \Pi_{i,\mathcal{J}}$, then $E \Pi_{i, \mathcal{J}\cup\{m,m^{\prime}\}}  \leq E \Pi_{i,\mathcal{J}\cup\{m\}}$ for all $m^{\prime}>m$. If $E \Pi_{i,\mathcal{J}\cup\{m,m^{\prime}\}} \geq E \Pi_{i,\mathcal{J}\cup\{m\}}$, then $E \Pi_{i, \mathcal{J}\cup\{m\}}  \geq E \Pi_{i,\mathcal{J}}$ for all $m < m^{\prime}$.
\end{lemma}

\begin{proof}
    See Appendix \ref{Lemma 2}.
\end{proof}

Lemma 2 implies that among $\{\vec{1},\vec{2},\cdots,\vec{J}\}$, the optimal assortment $\vec{m}$ is defined as
\begin{equation}
    m = 
\arg \max_j \quad \{\pi_{j}-E\Pi_{i,\overrightarrow{j-1}}\geq 0, \pi_{j+1}-E\Pi_{i,\vec{j}}<0\}.
\end{equation}

\begin{lemma}
    When $\vec{m}$ is optimal among $\{\vec{1},\vec{2},\cdots,\vec{J}\}$, $\vec{m}$ dominates all non-consecutive assortments.
\end{lemma}

\begin{proof}
    See Appendix \ref{Lemma 3}.
\end{proof}

\begin{theo}
\label{theo1}
    The firm's optimal assortment for consumer $i$ belongs to $\{\vec{1},\vec{2},\cdots,\vec{J}\}$.
\end{theo}

The above lemmas imply that the consumer-specific optimal assortment is $\mathcal{J}^*_i = \vec{j_i^*}$, where 
\begin{equation}
    j_i^* = 
\arg \max_j \quad \{\pi_{j}-E\Pi_{i,\overrightarrow{j-1}}\geq 0, \pi_{j+1}-E\Pi_{i,\vec{j}}<0\}.
\end{equation}

\subsection{A cutoff property}
\label{Sampling}

The firm can select optimal assortments solely from the foldable menu, as established by Theorem \ref{theo1}. When products exhibit distinct price-cost margins—specifically, when $\pi_1 > \pi_2 > ... > \pi_J > \pi_0 = 0$—the assortment $\vec{j_i^*}$ strictly dominates all other assortments. This dominance implies that the optimal assortment for each consumer is unique within the foldable menu, which in turn facilitates the recovery of unobserved choice set heterogeneity. This theoretical framework aligns with our approach in estimating cigarette demand in China's tobacco industry, where we categorize cigarettes into Tier I through V based on price, from high to low. Notably, high-priced cigarettes generally yield higher unit margins compared to lower-priced ones at both the wholesale and retail levels, reinforcing the industry-specific reality of non-identical price-cost margins across cigarette tiers.

We can rewrite the utility function as
\[
u_{ij} = \delta_{ij}  + \epsilon_{ij} = \gamma_{ij} - \alpha_i p_j + \epsilon_{ij}
\]
where $\delta_{ij}$ is decomposed into the utility from consumption $\gamma_{ij}$ and the disutility from paying the price, represented by $-\alpha_i p_j$. To determine the optimal assortment for consumer $i$, we calculate $J-1$ individual-specific cutoffs. Specifically, the optimal assortment is $\vec{j}$ if $\alpha_i \in (c^i_{j-1,j}, c^i_{j,j+1}]$, where $c^i_{0,1} = -\infty$ and $c^i_{J,J+1} = +\infty$.

In fact, the cutoff $c^i_{j-1,j}$ is the $\alpha_i$ such that    
\begin{equation}
\label{eq:cutoff}
    \pi_j = \sum_{j'=1}^{j-1} \pi_{j'} \frac{\exp(\gamma_{ij'} - \alpha_i p_{j'})}{1 + \sum_{j''=1}^{j-1}  \exp(\gamma_{ij''} - \alpha_i p_{j''})} = E\Pi_{i,\overrightarrow{j-1}}.
\end{equation}
This equation can be solved explicitly when there are only two products: $c^i_{1,2} = \frac{1}{p_1}[\gamma_{i1} - \log \frac{\pi_2}{\pi_1 - \pi_2}]$. Intuitively, consumer $i$ is more likely to face a simple assortment with just product 1 when (i) product 1 becomes more attractive or less expensive; (ii) product 1's profit margin is higher; (iii) product 2's profit margin is lower.  

\begin{lemma}
    $c^i_{j-1,j}$ is strictly increasing in $j$ if the products do not have identical price-cost margins.
\end{lemma}

\begin{proof}
  Let us arrange the products in descending order as $\pi_1 > \pi_2 > \ldots > \pi_J > \pi_0 = 0$, so that $\pi_j$ is strictly decreasing in $j$. Note that Equation \eqref{eq:cutoff} is equivalent to $\pi_j = \sum_{j'=1}^{j-1} (\pi_{j'} - \pi_j)\exp(\gamma_{ij'} - \alpha_i p_{j'})$.  For a fixed $\alpha_i$, $\sum_{j'=1}^{j-1} (\pi_{j'} - \pi_j)\exp(\gamma_{ij'} - \alpha_i p_{j'})$ is strictly increasing in $j$. For a fixed $j$, $\sum_{j'=1}^{j-1} (\pi_{j'} - \pi_j)\exp(\gamma_{ij'} - \alpha_i p_{j'})$ is clearly strictly decreasing in $\alpha_i$. Therefore, to establish $\pi_{j+1} = E\Pi_{i,\vec{j}}$, which is equivalent to $\pi_{j+1} = \sum_{j'=1}^{j} (\pi_{j'} - \pi_{j+1})\exp(\gamma_{ij'} - \alpha_i p_{j'})$, it is necessary that $c^i_{j,j+1} > c^i_{j-1,j}$.
\end{proof}

\begin{theo}
\label{theo:2}
    The firm's optimal assortment for consumer $i$ with $\alpha_i \in (c^i_{j-1,j},c^i_{j,j+1}]$ is $\vec{j}$. 
\end{theo}

\begin{proof}
    For consumer $i$ with $\alpha_i \in (c^i_{j-1,j},c^i_{j,j+1})$, the definitions of $c^i_{j-1,j}$ and $c^i_{j,j+1}$ imply that $E\Pi_{i,\overrightarrow{j-1}}  < E\Pi_{i,\vec{j}}$ and $E\Pi_{i,\overrightarrow{j+1}}  < E\Pi_{i,\vec{j}}$, respectively. Therefore, the optimal assortment is $\vec{j}$.     
\end{proof}

\section{Empirical Application}

In this section, we apply the foldable menu model to analyze China's tobacco industry. First, we describe our data sources and outline the definitions and constructions of the main variables.

\subsection{Data and Variables}

\subsubsection{Cigarette sales and market shares}

We have collected provincial-level data on cigarette sales across the five tiers from the 2011-2016 China Tobacco Yearbooks. This dataset includes information for each of China's 31 provincial-level administrative divisions. In total, it encompasses sales data for the five cigarette tiers across 31 provinces over a six-year period, resulting in 930 data points. Appendix \ref{Summary Statistics of Cigarette Sales} provides a detailed summary of the statistics related to cigarette tier sales.\footnote{We impute Shanghai's cigarette sales for 2012 as the average of sales in 2011 and 2013 due to a clear error in the original 2012 data.}

To account for the possibility of individuals choosing not to smoke, it is essential to estimate China's potential cigarette market size. This metric represents the hypothetical number of cigarettes consumed if the entire population aged 15 and above were smokers. According to data from the Global Market Information Database, the prevalence of smoking among adults in China was approximately 28.1\% in 2011, 28\% in 2012 and 2013, and 27.9\% from 2014 to 2016. To construct a measure of market share relative to this potential market size, we multiply the percentages of cigarette tier sales by the smoking prevalence rates in the  corresponding years. These calculated market shares serve as the observed market shares in the application, and Table \ref{tab:1.1} provides a summary of these observed market shares. Notably, Tiers I and II exhibit clear upward trajectories in market share, while Tiers IV and V show distinct downward trends. Tier III's market share experiences a significant increase from 2011 to 2012 and stabilizes thereafter.

\begin{table}[ht]\centering
\caption{Summary Statistics of Observed Market Shares (\%)}
\resizebox{.8\columnwidth}{!}{%
\label{tab:1.1}
\begin{tabular}{c|cccccc} \hline \hline
\diagbox[width=5em]{Tier}{Year}                    & 2011   & 2012   & 2013   & 2014   & 2015   & 2016   \\ \hline
\multirow{2}{*}{I} & 3.68 & 4.19 & 4.77 & 5.41 & 5.73 & 5.62 \\
                   & (1.88) & (1.54) & (1.74) & (1.94) & (2.27) & (2.54) \\ \hline
\multirow{2}{*}{II} & 1.77 & 2.14 & 2.47 & 2.88 & 3.32 & 3.71 \\
                   & (1.04) & (1.12) & (1.26) & (1.43) & (1.47) & (1.56) \\ \hline
\multirow{2}{*}{III} & 10.69 & 12.56 & 12.85 & 12.45 & 12.2 & 12.05 \\
                   & (2.93) & (2.63) & (2.62) & (2.31) & (2.39) & (2.54) \\ \hline
\multirow{2}{*}{IV} & 8.22 & 6.37 & 5.87 & 5.44 & 5.08 & 5.02 \\
                   & (2.66) & (2.08) & (2.13) & (2.04) & (2.17) & (2.4) \\ \hline
\multirow{2}{*}{V} & 3.74 & 2.75 & 2.04 & 1.73 & 1.58 & 1.51 \\
                   & (1.47) & (0.94) & (0.78) & (0.77) & (0.78) & (0.78) \\ \hline
N & 31 & 31 & 31 & 31 & 31 & 31 \\ \hline \hline
\multicolumn{7}{p{0.65\textwidth}}{\footnotesize Notes: Observed market shares are constructed by multiplying the percentages of cigarette tier sales by the smoking prevalence rates in the corresponding years. Values are the mean across 31 provinces. Standard deviations are in parentheses.}\\
\end{tabular}
}
\end{table}

\subsubsection{Cigarette Tier Prices and Profits}

In the demand estimation of the foldable menu model, it is crucial to have access to cigarette retail prices and wholesale profits. \cite{goodchild2018early} identify five brands in each cigarette tier to serve as a representative sample, selected based their significant market share within their respective tier. The study provides the average wholesale and retail prices for these selected cigarette brands for the years 2014 to 2016, as detailed in Table \ref{tab:1.2}.

\begin{table}[ht]\centering
\caption{Average Nominal Prices and Margins}
\resizebox{\columnwidth}{!}{%
\label{tab:1.2}
\begin{tabular}{ccccccc}
\hline \hline
\multirow{2}{*}{Tier} & \multicolumn{2}{c}{Wholesale Price} & \multicolumn{2}{c}{Retail Price} & \multicolumn{2}{c}{Wholesale Margin} \\
                      & 2011-2014       & 2015-2016      & 2011-2014        & 2015-2016        & 2011-2014         & 2015-2016        \\ \hline
I                     & 21.8             & 23   & 24.5            & 26.5         & 3.75              & 3.65             \\
II                    & 11.6             & 12.3              & 13              & 14         & 1.59              & 1.49             \\
III               & 8.3              & 8.8     
& 9.5             & 10         & 1.14              & 1.04             \\
IV               & 4.5              & 4.8      & 5               & 5.5        & 0.48              & 0.38             \\
V                & 2.3              & 2.4      & 2.5             & 3         & 0.17              & 0.07             \\ \hline \hline
\multicolumn{7}{p{0.8\textwidth}}{\footnotesize Notes: Prices are the average prices of each tier's five representative cigarette brands, as selected by \cite{goodchild2018early}. Prices change due to a tax increase in 2015.}\\
\end{tabular}
}
\end{table}

Cigarette prices rose between 2014 and 2015 due to an exogenous tax increase in May 2015. STMA permitted wholesale cigarette prices to align with the tax hike in May 2015. In contrast, during another exogenous tax increase in May 2009, STMA prohibited passing on the tax increase to consumers through higher cigarette prices.\footnote{STMA issued a document, entitled Notice of Adjusting Cigarettes Allocation Price (2009, No. 180), also effective May 1, 2009, which declared that wholesale cigarette prices should remain the same nationwide as before the excise tax adjustment. Given that the excise tax is collected at the producer and wholesale levels but not at the retail level, retail cigarette prices remain the same as wholesale cigarette prices remain the same.}  Therefore, we treat the nominal prices as constant in 2011-2014. These tax increase events highlight the government regulation on China's tobacco prices. Details of the tax increases are summarized in Appendix \ref{Cigarette Tax Increases}.

In this application, considering nominal price rigidity, we assume that provincial CNTCs endogenously select cigarette assortments to maximize wholesale after-tax profits. We maintain fixed costs, such as advertising and management, as constant regardless of cigarette sales. Consequently, after-tax profits depend on the sales volume and the unit after-tax price-cost margin for each cigarette tier. To compute these unit after-tax cigarette wholesale margins, we refer to information on cigarette taxes, China's cigarette pricing mechanism, and cigarette wholesale margin rates from \cite{gao2012can} and \cite{zheng2018tobacco}. The detailed calculation process for wholesale after-tax price-cost margins is provided in Appendix \ref{Wholesale Profits}. The last two columns of Table \ref{tab:1.2} summarize these unit after-tax cigarette wholesale margins. As highlighted in \cite{goodchild2018early}, the increase in the wholesale price of cigarettes corresponds to the change in the ad valorem excise tax rate, which increased from 5\% to 11\% in 2015. This suggests that the new specific excise tax of \yen 0.10/pack was absorbed into wholesale margins rather than being passed on to consumers. Consequently, wholesale after-tax price-cost margins for 2015 and 2016 are calculated by subtracting the new specific excise tax of \yen 0.10/pack from the values for 2011-2014.

One challenge in this analysis is that cigarette prices and profits are presented at the national level, while cigarette sales and market shares are broken down at the provincial level. Since the composition of cigarette tier sales varies across provinces and years, the weighted average cigarette tier prices and profits also fluctuate across provinces and years. Unfortunately, detailed data on provincial cigarette sales compositions are unavailable, so we assume that nominal prices and profits for cigarette tiers are uniform across all provinces. However, it is worth noting that this model can be readily adapted to account for varying compositions of cigarette tiers when more detailed datasets are accessible.

\subsubsection{CPI and Disposable Income}

The nominal prices and profits for cigarette tiers, as presented in Table \ref{tab:1.2}, are uniform throughout the country. However, there is variation in consumer price indexes (CPIs) across provinces and years, which is crucial for estimating price sensitivity parameters. To capture this variation and obtain real cigarette tier prices and profits, we have collected provincial-level CPIs from the 2011-2016 China Statistics Yearbooks. These CPIs allow us to adjust nominal retail prices and wholesale profits, converting them into real prices and profits.

To account for consumers' disposable income, which are observable demographic factors influencing cigarette purchasing behavior, we have collected disposable income data from the 2011-2016 Provincial Statistical Yearbooks for each province. These yearbooks provide information on both the overall average disposable income and the average income for each income quintile within each province and year. Following \cite{berry1995automobile}, we assume the income distribution to be lognormal. Subsequently, we estimate the mean and standard deviation of the lognormal distribution based on the provincial quintile data, as outlined in \cite{xiao2017welfare}.\footnote{We initiate the process by generating 10,000 individuals following a lognormal distribution with a mean of $a$ and a standard deviation of $b$. Formally, the natural logarithm of the annual income of individual $i$ is expressed as $ln(inc_i) = a + b \mu$, where $\mu \sim N(0, 1)$. Subsequently, we compute the mean of each quintile within the simulated sample and align it with the observed quintile averages. The sum of squares of the differences between these two values is then calculated. The optimal values for $a$ and $b$ are determined by minimizing this summation for each province in each year.} 

Similar to cigarette prices, incomes reported in the Provincial Statistical Yearbooks are presented in nominal terms. We also use the CPIs for each province and year to adjust nominal incomes for inflation. A summary of the CPIs, nominal provincial mean disposable incomes, and real provincial mean disposable incomes for these years can be found in Table \ref{tab:1.3}.

\begin{table}[ht]\centering
\caption{Summary Statistics of CPI and Disposable Income}
\resizebox{.8\columnwidth}{!}{%
\label{tab:1.3}
\begin{tabular}{c|cccccc} \hline \hline
\diagbox[width=7em]{Variable}{Year}                    & 2011   & 2012   & 2013   & 2014   & 2015   & 2016   \\ \hline
\multirow{2}{*}{\begin{tabular}[c]{@{}c@{}}Overall\\  CPI\end{tabular}}        & 105.49 & 108.38 & 111.44 & 113.69 & 115.34 & 117.43 \\
                                                                                 & (0.38)   & (0.65)   & (1.03)   & (1.19)   & (1.39)   & (1.59)   \\ \hline 
\multirow{2}{*}{\begin{tabular}[c]{@{}c@{}}Nominal Mean \\ Income (\yen 10,000)\end{tabular}}        & 2.06 & 2.32 & 2.53 & 2.75 & 2.99 & 3.23 \\
                                                                                 & (5.36)  & (5.84)  & (6.45)  & (7.01)  & (7.56)  & (8.26)  \\ \hline 
\multirow{2}{*}{\begin{tabular}[c]{@{}c@{}}Real Mean \\ Income (\yen 10,000)\end{tabular}}           & 1.95 & 2.14 & 2.27 & 2.42 & 2.59 & 2.75 \\
                                                                                 & (5.11)  & (5.43)  & (5.84)  & (6.21)  & (6.54)  & (6.94) \\
\hline 
N & 31 & 31 & 31 & 31 & 31 & 31 \\ \hline \hline 
\multicolumn{7}{p{0.78\textwidth}}{\footnotesize Notes: Data are provided by the 2011-2016 China Statistics Yearbooks and Provincial Statistics Yearbooks. The base year for CPIs is 2010. Values are the mean across 31 provinces. Standard deviations are in parentheses.}\\     
\end{tabular}
}
\end{table}

Real disposable income plays a crucial role in our model due to its positive correlation with the market shares of Grade A cigarettes, which include Tiers I and II. To evaluate this correlation, we conducted an Ordinary Least Squares (OLS) regression, linking the observed market shares of Grade A cigarettes to the logarithm of real mean disposable income. The results of this regression are shown in Table \ref{tab:1.4}. The income coefficient in the regression is positive and statistically significant at the 1\% level. On average, a 1\% increase in real mean disposable income is associated with a 5.6 percentage point increase in the market share of Grade A cigarettes.

\begin{table}[ht]\centering
\def\sym#1{\ifmmode^{#1}\else\(^{#1}\)\fi}
\caption{OLS Regression Results}
\label{tab:1.4}
\begin{tabular}{l*{1}{c}}
\\ \hline\hline
            &\multicolumn{1}{c}{Observed market share of Grade A cigarettes}\\
\hline
Log mean disposable income      &       0.056\sym{***}\\
            &      (0.008)         \\
[1em]
Constant      &      -0.488\sym{***}\\
            &      (0.075)         \\
\hline
\(N\)       &         186         \\
\hline\hline
\multicolumn{2}{l}{\footnotesize Robust standard errors are in parentheses.}\\
\multicolumn{2}{l}{\footnotesize \sym{*} \(p<0.05\), \sym{**} \(p<0.01\), \sym{***} \(p<0.001\)}\\
\end{tabular}
\end{table}

However, it is important to note that this regression result may overstate the effect of income on the market shares of Grade A cigarettes. This overestimation arises because high-income consumers are more likely to be offered limited assortments featuring only high-priced cigarettes, consistent with CNTC's Premiumization Strategy. As a result, endogenous assortments have a larger positive impact on the market shares of Grade A cigarettes in high-income provinces compared to low-income provinces. To disentangle the effects of income from the influence of endogenous assortments, the proposed foldable menu model is necessary for demand estimation.

\subsection{Utility Function}
\label{utility}

In estimating the foldable menu model using the actual dataset, we define the utility function for consumer $i$, who selects a tier $j$ cigarette in province $l$ and year $t$, as follows:
\begin{equation}
\label{eq:utility}
u_{ijlt} = \delta_{ijlt} + \epsilon_{ijlt} = \gamma_{jlt} - \alpha_{i} p_{jlt} +\epsilon_{ijlt},
\end{equation}
where $\gamma_{jlt} = \xi_{jl} + \Delta \xi_{jlt}$ represents the fixed utility of tier $j$ cigarettes in province $l$ and year $t$; $\xi_{jl}$ captures the time-invariant fixed utility, and $\Delta \xi_{jlt}$ represents the time-variant random demand shocks. $- \alpha_{i} p_{jlt} +\epsilon_{ijlt}$ represents individual deviation from the fixed utilities with $\alpha_{i} = \theta/inc_{i}$.\footnote{This functional form for the interaction between income and price, also used in \cite{berry1999voluntary}, can be derived as a first-order Taylor series approximation to the Cobb–Douglas utility function originally used in \cite{berry1995automobile}.} Specifically, $inc_{i}$ represents the simulated real income in units of \yen 10,000.\footnote{The term $\alpha_{i} p_{jlt} = \theta/inc_{i} \times p_{jlt}$ remains unchanged whether we use real income and real price or nominal income and nominal price. In this application, we use real terms since we also use real profits in the estimation.} In essence, the varying market shares of cigarette tiers across different provinces and years are determined by several factors, including tier fixed utilities, real retail prices, real wholesale profits, and the distribution of real disposable incomes. Notably, consumers in wealthier provinces tend to purchase more high-priced cigarettes due to their lower price sensitivity and a reduced likelihood of encountering low-priced cigarettes in their local stores, a phenomenon attributed to the foldable menu model.

Note that in Equation \eqref{eq:utility}, the term $\gamma_{jlt}$ is not individual-specific, meaning we assume that all consumers share the same fixed utility from tier $j$ in province $l$ and year $t$. While this assumption may be strong, it significantly reduces the computational burden of determining the optimal assortments for all consumers in a province-year. Instead of calculating individual cutoffs for each consumer, we only need to compute one set of cutoffs for each province-year. Once the cutoffs for province $l$ and year $t$ are determined, Theorem \ref{theo:2} states that the firm's optimal assortment for all consumers with $\alpha_{i} \in (c^{lt}_{j-1,j}, c^{lt}_{j,j+1}]$ is $\vec{j}$.

This cutoff property facilitates estimation. While it is standard practice in the numerical search process to use the same set of random draws to calculate market shares so that they are smooth functions of the parameters, the same random draw may face different assortments when the parameters change, leading non-smooth market shares. To restock the smoothness, we adapt the standard simulation approach using the cutoff property. In particular, we use the same set of random draws in each interval. 

Suppose the density of $\alpha_i$, denoted as $\phi$, has full support over the interval $[a,b]$, yielding a conditional density of $\frac{\phi(\alpha)}{\Phi(b)-\Phi(a)}$. First, we take a random sample from a standard uniform distribution $u_{i}$. Second, to generate a random sample of consumers on this interval for each set of parameters, we use this conditional density's quantile function\footnote{The quantile function is defined by $u=\int_{a}^{\alpha}\frac{\phi(v)}{\Phi(b)-\Phi(a)}dv = \frac{\Phi(\alpha)-\Phi(a)}{\Phi(b)-\Phi(a)}$.}
\[
\alpha_{i}=\Phi^{-1}\Big(\Phi(a)+u_{i}[\Phi(b)-\Phi(a)]\Big).
\]
We can then calculate the overall choice probabilities of a group of consumers facing the same assortment. Specifically, the market share of product $j$ in province $l$ and year $t$ is 
\begin{equation}
\label{eq:predicted shares}
\sum_{j'=j}^J
\underbrace{[\Phi(c^{lt}_{j',j'+1})-\Phi(c^{lt}_{j'-1,j'})]}_{\text{share of assortment } \vec{j'}}
\left[
\underbrace{\frac{1}{N_{s}} \sum_{i=1}^{N_s} \frac{\exp(\gamma_{jlt} - \alpha_{i} p_{jlt})}{1+\sum_{j''=1}^{j'} \exp(\gamma_{j''lt} - \alpha_{i} p_{j''lt})}}_{\text{overall market share among consumers facing } \vec{j'}} 
\right],
\end{equation}
which is smooth in the parameters because $\alpha_i$ is.

In our study, nominal cigarette prices are treated as fixed and determined by the STMA at the introduction of cigarette brands. Fluctuations in real cigarette prices across different provinces and years are primarily influenced by changes in the CPIs. Importantly, these variations are assumed to be uncorrelated with random demand shocks. For instance, a random demand shock in a market may arise from an anti-high-consumption campaign, leading to a decrease in consumers' preferences for high-priced cigarettes. This type of demand shock is clearly unrelated to nominal cigarette prices and the CPIs. In the subsequent Monte Carlo experiment and empirical application, the identification assumption posits that these demand shocks are uncorrelated with real cigarette prices.

\subsection{Monte Carlo Experiment}
\label{Monte Carlo}

In this section, we conduct a Monte Carlo experiment to validate our estimation strategy.
\subsubsection{Setup} \label{secDGP}
We set the number of markets $T = 186$, aligning with the China tobacco dataset. The vector $\xi = \{2;1.5;1.2;1;0.8\}$ represents the constant tier fixed utilities. The nominal price vector is $P = \{3.6;2.4;1.6;1.2;1\}$ and the nominal profit vector is $\Pi = \{2.5;1.8;1.2;1;0.8\}$. 

To simulate the Consumer Price Index (CPI) in each market, we use the formula 
\[
CPI_t = 1 + 0.2 \times \epsilon_{1t}.
\] 
The inflation multiplier $I_t = 1/CPI_t$ is employed to deflate nominal prices and profits to obtain the real price vector $p_t$ and real profit vector $\pi_t$. Income is log-normal distributed. The mean and standard deviation of log income are generated using 
\[
\mu_t = 1 + 0.1 \times \epsilon_{2t},
\]
\[
\sigma_t = 0.5 + 0.1 \times \epsilon_{3t}.
\]
Here, $\epsilon_1, \epsilon_2, \epsilon_3 \sim U(0,1)$. Furthermore, unobservable demand shock $\Delta \xi_{jt}$ is simulated using $0.3 \times \rho_{jt}$, where $\rho \sim N(0,1)$.

Individual utility is
\[
u_{ijt} = \gamma_{jt} - \alpha_{i} p_{jt} + \varepsilon_{ijt} = \xi_{j} + \Delta \xi_{jt} - \alpha_{i} p_{jt} + \varepsilon_{ijt},
\]
where $\alpha_{i} = \theta/inc_{i}$ with $\theta = 2$, so that all consumers have positive $\alpha$. 

\subsubsection{Market Share Calculation}
For each market $t$, given $\gamma_t$, $p_t$, and $\pi_t$, we can calculate the cutoff $c^t_{j-1,j}$ using $\pi_{jt}  = E\Pi_{i,\overrightarrow{j-1},t}$ for all $j > 1$. Following the sampling method outlined in Section \ref{utility}, we generate random uniform draws $u_{i}$ with a sample size of 2000, and then use the conditional density's quantile function to generate a random sample of $\alpha_i$ on the interval $(c^t_{j-1,j},c^t_{j,j+1}]$ for each set of parameters. With these values, we can calculate the market share of product $j$ in market $t$ using
\[
s(\gamma_t, \theta) = \sum_{j'=j}^J
\underbrace{[\Phi(c^t_{j',j'+1})-\Phi(c^t_{j'-1,j'})]}_{\text{share of assortment } \vec{j'}}
\left[
\underbrace{\frac{1}{N_{s}} \sum_{i=1}^{N_s} \frac{\exp(\gamma_{jt} - \alpha_{i} p_{jt})}{1+\sum_{j''=1}^{j'} \exp(\gamma_{j''t} - \alpha_{i} p_{j''t})}}_{\text{overall market share among consumers facing } \vec{j'}} 
\right].
\]

When solving for $\pi_{jt} = E\Pi_{i,\overrightarrow{j-1},t}$, it is possible to encounter $c^t_{j-1,j} \leq 0$ for any $j>1$. However, our model settings dictate that $\alpha$ values must be positive. Therefore, when facing the conditions $c^t_{j-1,j} \leq 0$ and $c^t_{j,j+1} > 0$, the implicit assumption is that the shortest assortment is $\vec{j}$ rather than $\overrightarrow{j-1}$. As a result, consumers are grouped starting from $c^t_{j,j+1}$, and in the above market share expression, we have $\Phi(c^t_{j-1,j}) = 0$. Under the above model setting, the mean market share is 13.7\% and the standard deviation is 11\%. The minimum is 0.9\% and the maximum is 44.7\%.

\subsubsection{Estimation}

Since individual choices are unobservable, we cannot estimate the foldable menu model using maximum likelihood estimation (MLE). Instead, we estimate the model using GMM, similar to the approach in \cite{berry1995automobile}. The first step involves deriving fixed utilities $\gamma_{jt}$ through contraction mapping. For any given $\theta$, we have knowledge of the distribution of $\alpha_{i}$. We initialize $\gamma_{jt}$ as $log(s_{jt}^{obs}) - log(s_{0t}^{obs})$. Given $\gamma_t$ and observed $p_t$ and $\pi_t$, we can calculate the cutoffs for $\alpha_{i}$ and then determine predicted market shares $s(\gamma_t, \theta)$ using the same method as described in the preceding section. We iteratively update $\gamma$ using the formula $\gamma^{\tau + 1} = \gamma^{\tau} + \ln(s^{obs}) - \ln(s(\gamma^{\tau}, \theta))$ until the condition $||\gamma^{\tau + 1} - \gamma^{\tau}|| < 10^{-6}$ is satisfied.\footnote{The contraction mapping fails only when $\theta$ is too small, causing the likelihood to offer the least profitable product smaller than its market share in the actual data.}


After obtaining the converged $\gamma_{jt}(\theta)$, we perform an OLS regression of $\gamma_{jt}(\theta)$ on cigarette tier dummies to obtain estimates for $\xi_j(\theta)$. Subsequently, we derive $\Delta \xi_{jt}(\theta) = \gamma_{jt}(\theta) - \xi_j(\theta)$. Real prices, denoted as $p$, serve as the instrumental variable, as they are uncorrelated with the unobserved product heterogeneity $\Delta \xi(\theta)$. The single moment is then calculated as $\Delta \xi(\theta)' p$, and the objective function is defined as the square of this moment. The search for the optimal $\theta$ involves minimizing the objective function value. The following Figure \ref{fig:1.2} illustrates the estimation process.\footnote{Alternatively, we could adopt the MPEC approach outlined in \citet{dube2012improving}, which eliminates the need to repeatedly calculating the cutoffs.}



\begin{figure}[hbtp]
    \centering
\begin{tikzpicture}[auto, every node/.style={font=\small}]
    \node (initTheta) [rectangle, draw, minimum width=4cm, text centered, rounded corners] {Initiate $\theta$};
    \node (calcDelta) [rectangle, draw, below=1.5cm of initTheta, minimum width=4cm, text centered] {Calculate $\gamma(\theta)$};
    \node (estimateDelta) [rectangle, draw, below=1cm of calcDelta, minimum width=4cm, text centered] {Reg. $\gamma(\theta)$ on tier dummies, predict $\Delta \xi_{jt}(\theta)$};
    \node (evalObjective) [rectangle, draw, below=1cm of estimateDelta, minimum width=4cm, text centered, rounded corners] {Evaluate GMM Objective Function};
    
    \begin{scope}[xshift=6cm, node distance=1cm]
        \node (innerStart) [rectangle, draw, minimum width=4cm, text centered] {Initial $\gamma(\theta)$};
        \node (calcCutoffs) [rectangle, draw, below of=innerStart, minimum width=4cm, text centered] {Calculate cutoffs};
        \node (calcS) [rectangle, draw, below of=calcCutoffs, minimum width=4cm, text centered] {Calculate $s(\gamma, \theta)$};
        \node (updateDeltaM) [rectangle, draw, below of=calcS, minimum width=4cm, text centered] {Update $\gamma(\theta)$};
        
        \draw[->,dashed] (innerStart) -- (calcCutoffs);
        \draw[->,dashed] (calcCutoffs) -- (calcS);
        \draw[->,dashed] (calcS) -- (updateDeltaM);
        \draw[->,dashed] (updateDeltaM.east) -- ++(1,0) |- (innerStart.east);
    \end{scope}
    
    \draw[->] (initTheta) -- (calcDelta);
    \draw[->] (calcDelta) -- (estimateDelta);
    \draw[->] (estimateDelta) -- (evalObjective);
    
    \draw[decorate,decoration={brace,amplitude=10pt,raise=4pt},yshift=0pt]
    ([xshift=-0.1cm]updateDeltaM.west) -- ([xshift=-0.1cm]innerStart.west) node [black,midway,xshift=0.8cm] {};

    \draw[->] ([xshift=-.65cm,yshift=-0.435cm]calcCutoffs.west) -- (calcDelta.east);

\end{tikzpicture}

    \caption{Estimation Process}
    \label{fig:1.2}
\end{figure}
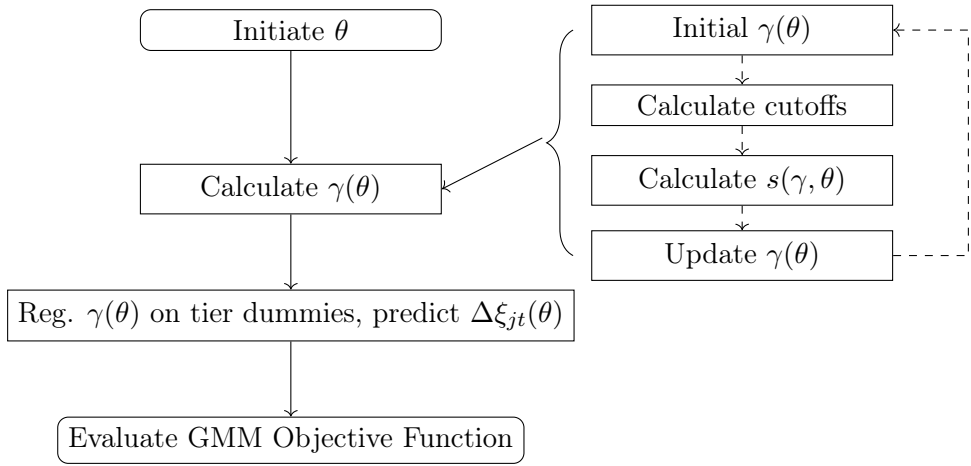

\begin{remark}
In the Monte Carlo experiment and the empirical application presented in the following section, we assume exogenous variations in CPI and that the real price $p$ is uncorrelated with $\Delta \xi$. The method can be easily extended to account for endogenous prices by applying GMM with appropriate instruments for the real price.
\end{remark}

To demonstrate the estimation bias that arises when ignoring assortment endogeneity, we also estimate the logit model using the samples generated in Section \ref{secDGP}, but under the assumption of consistent availability of all products (hereafter referred to as the standard logit model). Specifically, we simulate a random vector $v_i \sim N(0,1)$ with a size of 10,000 (5 times the size of $u_i$, since we use $u_i$ in each of the 5 intervals). Subsequently, we generate individual income in market $t$ as $inc_{i} = \exp(\mu_t + \sigma_t v_i)$ and calculate $\alpha_{i} = \theta/inc_{i}$. 



\begin{table}[ht]\centering
\caption{Summary Statistics}
\label{tab:1.13}
\resizebox{.8\columnwidth}{!}{%
\begin{tabular}{cccccccc}\hline\hline
 & & $\xi_1$      & $\xi_2$      & $\xi_3$       & $\xi_4$       & $\xi_5$        & $\theta$         \\ \hline
\multirow{2}{*}{Foldable Menu} & Mean  & 2.044  &  1.545   & 1.218  &  1.019   & 0.799 & 2.040   \\
& SD   &  0.195  &  0.197   & 0.113  &  0.043  &  0.077 &   0.192 \\ \hline
\multirow{2}{*}{Standard Logit} & Mean    & 0.122  & -0.072 &  -1.160  & -1.634 &  -2.446 & 0.239 \\
& SD   &  0.128  &  0.089  &  0.069   & 0.055  &  0.048 & 0.096 \\ \hline
& True      & 2    & 1.5    & 1.2    & 1 & 0.8   & 2 \\ \hline\hline
\end{tabular}
}
\end{table}

Table \ref{tab:1.13} summarizes the estimates from 20 simulated samples. Compared to the foldable menu model, the standard logit model tends to underestimate \( \theta \). This is because a lower \( \theta \) has a larger positive effect on the utility derived from higher-priced products than from lower-priced ones, which leads to an increase in the sales of high-tier cigarettes. Another factor that contributes to the increase in the sales of high-tier cigarettes is endogenous assortment. The standard logit model misattributes the effect of endogenous assortment to a lower \( \theta \), resulting in the underestimation of \( \theta \). Moreover, a lower \( \theta \) implies a smaller disutility from price. To align the predicted market shares with the observed ones, the standard logit model also underestimates \( \gamma \) and thus $\xi$ compared to the foldable menu model.

\section{Empirical Results}

The estimation process for the empirical application follows the same approach as in the Monte Carlo experiment. As detailed in Section \ref{Monte Carlo}, we successfully retrieve the pre-set model parameters using the estimation strategy. The empirical results are presented in Table \ref{tab:1.5}, showcasing estimation outcomes from both the foldable menu model and the standard logit model, the latter of which does not account for choice set heterogeneity.

\begin{table}[ht]\centering
\caption{Estimation Results}
\label{tab:1.5}
\setlength{\tabcolsep}{8mm}{
\begin{tabular}{ccccc}\\ \hline \hline
Parameters & \multicolumn{2}{c}{Foldable Menu} & \multicolumn{2}{c}{Standard Logit} \\ \cline{2-3}  \cline{4-5} 
           & Coef.            & SE/SD           & Coef.                  & SE/SD                 \\ \hline
$\theta$           &     2.012 &       0.451 & 1.025                       &      0.861 \\ \hline
$\xi_1$            &    0.660              &      0.446          &                       -0.694 &       0.372               \\
$\xi_2$            &    -1.265 &    0.540 &                       -2.187 &      0.492                \\
$\xi_3$            &      -0.054 &     0.312     &                      -0.862  &         0.230 \\
$\xi_4$            &      -1.027 &    0.431 &                       -2.055 &       0.452              \\
$\xi_5$            &     -1.173             &     0.507          &                       -3.368 &       0.490               \\ \hline \hline
\multicolumn{5}{p{.8\textwidth}}{\footnotesize Notes: We provide estimates for $\theta$ along with their bootstrap standard errors. Additionally, $\xi_j$ denotes the mean values of $\xi_{jl}$ across all 31 provinces in China, and we include the corresponding standard deviations.}\\ 
\end{tabular}
}
\end{table}

The parameter $\theta$ represents the point estimate using the true dataset, and we report the associated bootstrap standard error based on 20 bootstrap resamples. Upon obtaining the $\theta$ estimate, we invert the market share equation to derive $\gamma_{jlt}$ and subsequently conduct an OLS regression to obtain $\xi_{jl}$. We allow the time-invariant fixed utility to vary across different provinces to accommodate the varying preferences for cigarette tiers in different provinces. The resulting $\xi_j$ represents the mean values of $\xi_{jl}$ across all 31 provinces in China, and we include the corresponding standard deviations.

In the foldable menu model, all coefficients exhibit the expected signs and scales. Notably, fixed utility shows a monotonic relationship with cigarette tier prices, except for Tier II, which, on average, has a lower fixed utility than Tiers III to V. In contrast, the standard logit model underestimates the coefficients for price sensitivity and fixed utilities. These findings align closely with those from the Monte Carlo experiment. The underestimation of price sensitivity coefficients leads to lower estimates of price elasticity, as discussed in the subsequent section.

\subsection{Price Elasticities}

Calculating demand elasticities is a standard practice in the aftermath of demand estimation. We have compiled a summary of the own-price and cross-price demand elasticities for each of the five cigarette tiers. These elasticities measure the percent change in sales for a specific cigarette tier when its price increases by 1\%, while holding other tier prices constant. We report weighted average elasticities using the ratio of each market's total cigarette sales over the total sales across all provinces and years as weights.\footnote{Similarly, we report all subsequent elasticities and percentage changes in weighted average values.} Figure \ref{fig:compelas} provides a summary of these estimates using both the foldable menu model and the standard logit model.

\begin{figure}[ht]
    \centering
    \includegraphics[scale=0.66]{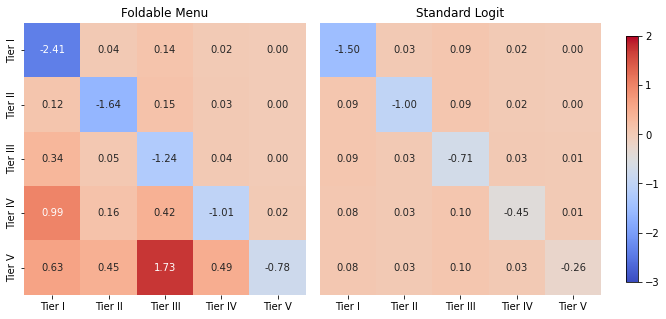}    
    \caption{Comparison of Price Elasticities. Notes: The columns for Tiers I to V represent cross (own) elasticities when the price of a specific tier increases by 1\% while holding other tier prices constant. The left panel includes the values estimated using the foldable menu model, while the right panel includes the values estimated using the standard logit model.}
    \label{fig:compelas}
\end{figure}

In comparison, the standard logit model—which overlooks choice set heterogeneity—tends to underestimate both own-price and cross-price elasticities for most cigarette tiers, except for the effect of a price increase in Tier V on the sales of Tier III. Another notable difference is that, under the foldable menu model, a price increase in higher-tier cigarettes prompts consumers to switch to lower-tier products, rather than distributing demand uniformly across tiers, as observed in the standard logit model. The differences in elasticities arise from two key factors: first, the parameter estimates differ between the foldable menu model and the standard logit model, with the latter underestimating elasticities by misattributing limited product substitutions—due to the constrained availability of low-priced cigarettes—as consumer insensitivity to price changes. Second, the foldable menu model allows for endogenous adjustments in product assortments in response to price changes. Specifically, when the price of Tier $j$ ($j < V$) rises, the cutoffs $c_{j,j+1}, \ldots, c_{IV,V}$ decrease, allowing more consumers access to lower-priced cigarettes, as outlined in Equation \eqref{eq:cutoff}. A detailed decomposition of these effects is provided in Appendix \ref{Decomposition of Elasticities}.

As a result, under the foldable menu model, when prices increase by 1\% across all tiers, Figure \ref{fig:allelas} shows a decline in the sales of Tiers I–III, while sales of Tiers IV and V rise. This contrasts sharply with the standard logit model, where a uniform 1\% price increase leads to declining sales across all tiers. Furthermore, the overall price elasticity—reflecting the percentage change in total cigarette demand when prices rise by 1\% across all tiers—is slightly larger in the standard logit model (-0.61) compared to the foldable menu model (-0.579).

\begin{figure}[ht]
    \centering
    \includegraphics[scale=0.3]{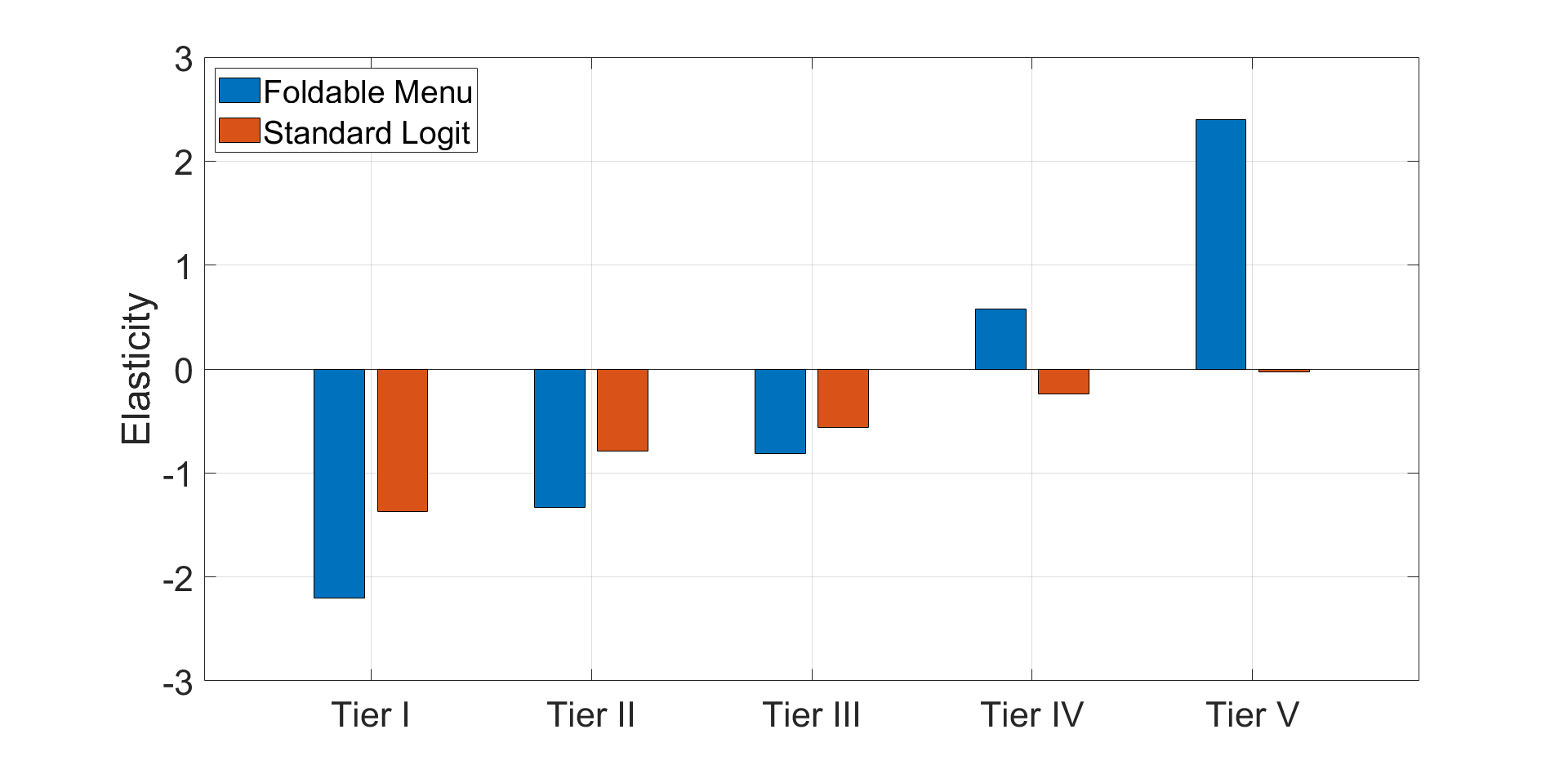}   
    \caption{Percentage Change in Sales for a 1\% Increase in All Prices}
    \label{fig:allelas}
\end{figure}

The price elasticities estimated using the foldable menu model appear relatively high compared to prior studies that assume cigarettes to be homogeneous, as seen in \cite{hu2002effects}, \cite{lance2004cigarette}, \cite{bishop2007chinese}, and \cite{chen2016quantity}. These studies tend to overlook substitutions that occur among different cigarette tiers. Furthermore, the estimated price elasticities using the foldable menu model also stand out as notably higher when compared to the price elasticities of individual cigarette brands, as indicated in \cite{liu2015smokers}. This disparity can be attributed to the fact that the price changes analyzed in this study occurred at the cigarette tier level, meaning that prices for all cigarette brands within the same tier increased by the same percentage. Consequently, these price changes have a more substantial impact on sales than when only the price of a single cigarette brand undergoes an increase.

A relevant study that estimated cigarette demand elasticities at the tier level is \cite{li2016heterogeneous}, which used individual survey data from China, spanning 2006-2009. In their analysis, the authors categorized cigarettes into four tiers by combining Tiers I and II. Their estimated own price elasticities are -2.506, -0.805, -0.39, and -0.28 for Tiers I\&II, III, IV, and V, respectively. These values (except for Tiers I\&II) closely align with the estimated elasticities obtained using the standard logit model. However, their analysis does not consider choice set heterogeneity, potentially leading to an underestimation of price elasticities, as discussed earlier.

Recovering essential unobservable variables from firms' profit maximization conditions offers a more accurate means of estimating price elasticities compared to models that overlook the challenge posed by unobservable variables. To illustrate this point, consider \cite{d2019automobile}, which recovers unobservable auto sale prices to investigate price discrimination. Their findings reveal that assuming all consumers pay the posted prices for autos (the uniform pricing model in their paper) resulted in an underestimation of price sensitivities and price elasticities. This discrepancy occurred because the assumption overlooks price negotiations and cash discounts, and thus overstates the actual prices paid by consumers in their study. Similarly, in this context, we have observed that assuming full availability of all cigarette tiers leads to an underestimation of price sensitivities and price elasticities. This is due to the erroneous attribution of limited product substitutions, arising from the foldable menu, as lower price elasticities. 

\subsection{Assortment Distribution}

Figure \ref{fig:1.1} illustrates the relationship between income and cigarette assortments. In each subfigure, the x-axis represents the mean of log real income in \yen 10,000, while the y-axis denotes the percentage of various assortments. It is evident that as income increases, the percentage of the full assortment decreases, while the percentage of assortment \{I,II,III\} increases, indicating that consumers in high-income provinces are more likely to encounter limited assortments compared to those in low-income provinces. The percentages of assortments \{I\}, \{I,II\}, and \{I,II,III,IV\} show no clear correlation with income. These assortment distributions help explain the counterfactual results based on the assumption of full availability of all cigarette tiers in the following section.

\begin{figure}[ht]
\centering
\begin{subfigure}{0.49\textwidth}
    \includegraphics[width=\textwidth]{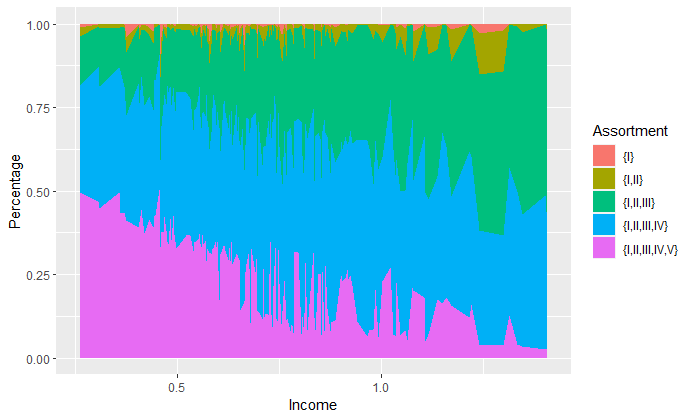}
    \caption{All Provinces and All Years}
    \label{fig:first}
\end{subfigure}
\hfill
\begin{subfigure}{0.49\textwidth}
    \includegraphics[width=\textwidth]{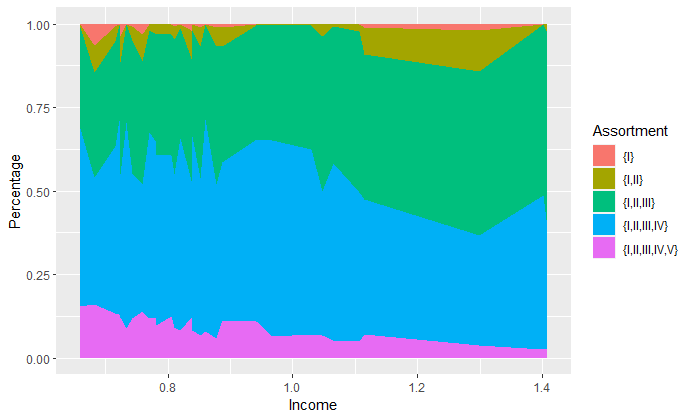}
    \caption{All Provinces and 2016}
    \label{fig:second}
\end{subfigure}
\hfill
\begin{subfigure}{0.49\textwidth}
    \includegraphics[width=\textwidth]{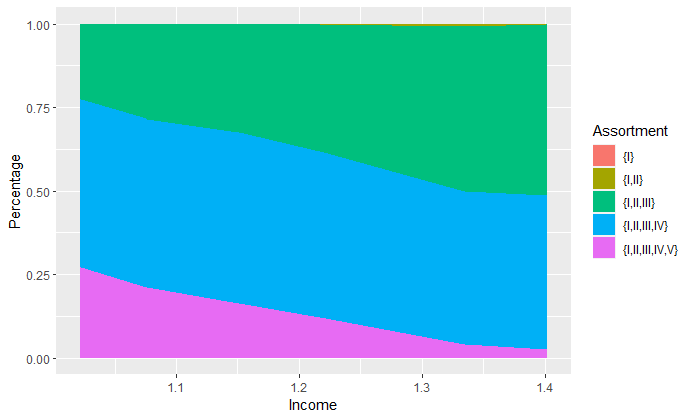}
    \caption{Beijing and All Years}
    \label{fig:third}
\end{subfigure}
\hfill
\begin{subfigure}{0.49\textwidth}
    \includegraphics[width=\textwidth]{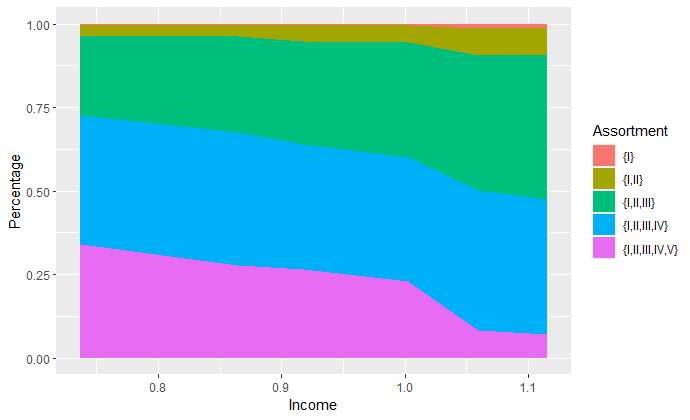}
    \caption{Jiangsu and All Years}
    \label{fig:fourth}
\end{subfigure}
        
\caption{Assortment Distribution}
\label{fig:1.1}
\end{figure}

Panel (a) depicts the assortment distribution over all provinces and all six years. As income increases, the percentage of the full assortment decreases from 49.5\% to 2.5\%, while the percentage of assortment \{I,II,III\} increases from 14.7\% to 58.6\%. Panel (b) illustrates the assortment distribution for all provinces in 2016. Similarly, as income increases, the percentage of the full assortment decreases from 15.5\% to 2.5\%, while the percentage of assortment \{I,II,III\} increases from 30.2\% to 58.6\%. Panel (c) showcases the assortment distribution in Beijing across all six years. Income strictly increases from 2011 to 2016. As income increases, the percentage of the full assortment decreases from 27.1\% to 2.7\%, while the percentage of assortment \{I,II,III\} increases from 22.4\% to 51.3\%. Lastly, Panel (d) details the assortment distribution in Jiangsu across all six years. Same as in Beijing, income strictly increases from 2011 to 2016. As income increases, the percentage of the full assortment decreases from 34.2\% to 7\%, while the percentage of assortment \{I,II,III\} increases from 23.4\% to 43.5\%.

\subsection{Counterfactual Experiments}

\subsubsection{Policy Implications}

Tobacco taxes are widely recognized as one of the most effective tools in tobacco control policies across various countries, as they reduce consumer demand for cigarettes through price increases while also generating tax revenue. Before implementing any tax increase policy, it is essential to thoroughly evaluate the potential implications for cigarette sales and tax revenue. Given that the misspecified standard logit model leads to biased estimates of parameters and price elasticities, it is crucial to incorporate endogenous assortment when assessing the implications of tax policy.

Following the May 2015 tax adjustment, the weighted tax share, represented as a percentage of the retail price, increased from 52\% in 2014 to 56\% in 2015 and 2016. Furthermore, the weighted excise tax share, also as a percentage of the retail price, rose from 31\% in 2014 to 34\% in 2015. However, both of these figures remain below the WHO recommended standards of 75\% and 70\%, respectively \citep{zheng2018tobacco}. To work towards achieving these WHO-recommended standards, a viable approach involves gradually increasing the ad valorem excise tax rate. This section simulates ad valorem excise tax increases of 5\%, 10\%, 15\%, and 20\%. For simplicity, we assume that the additional excise tax is imposed at the retail level, utilizing the current retail price as the tax base, and this does not impact existing after-tax unit wholesale margins.

Table \ref{tab:1.7} illustrates the impact of tax increases on sales for each cigarette tier under both models. These results align with earlier findings when prices for all tiers increase by 1\%. The foldable menu model observes a decrease in the sales of Tiers I to III, while the sales of Tiers IV and V increase. Conversely, the standard logit model shows a decrease in sales for all tiers, with a slight larger effect on total cigarette sales than under the foldable menu model.

\begin{table}[ht]\centering
\caption{Effect of Additional Tax on Cigarette Sales} 
\label{tab:1.7}
\resizebox{\columnwidth}{!}{%
\begin{tabular}{c|cccccc|cccccc} \hline \hline
       & \multicolumn{6}{c|}{Foldable Menu}          & \multicolumn{6}{c}{Standard Logit} \\ \hline 
       & Tier I & Tier II & Tier III & Tier IV & Tier V & All & Tier I  & Tier II  & Tier III & Tier IV & Tier V & All \\ \hline 
5\% & -10.35  & -6.40  & -4.01  & 2.69  & 12.09  & -2.83 & -6.57  & -3.85 & -2.78  & -1.21   & -0.16  & -2.99 \\
10\% & -19.43  & -12.37 & -7.99   & 4.81 & 24.32   & -5.52 & -12.62  & -7.54  & -5.49 & -2.43  & -0.36 & -5.86 \\
15\% & -27.44  & -17.91  & -11.90 & 6.39 & 36.58  & -8.07 & -18.20   & -11.07   & -8.13 & -3.68  & -0.61 & -8.61 \\
20\% & -34.51  & -23.07  & -15.71  & 7.50 & 48.85  & -10.48 & -23.36   & -14.45   & -10.70  & -4.93 & -0.91 & -11.25 
\\ \hline \hline
\multicolumn{13}{p{1.2\textwidth}}{\footnotesize Notes: Each row represents the effects of a tax increase on the percentage change in sales of each tier, as well as the percentage change in sales of all tiers. The left panel includes the values estimated using the foldable menu model, while the right panel includes the values estimated using the standard logit model.}\\ 
\end{tabular}%
}
\end{table}

Table \ref{tab:1.8} delves into the impact of tax increases on tax revenue under both models. As noted in Table \ref{tab:1.7}, the foldable menu model sees a more pronounced drop in sales for Tiers I to III, resulting in a smaller positive effect on tax revenue from these tiers. The effect on tax revenue from Tier I is even negative under the foldable menu model. Notably, due to the significantly higher unit tax on high-price cigarettes compared to low-price cigarettes, the standard logit model overstates the effect of additional excise tax on overall tax revenue by 58\% to 74\%.

\begin{table}[ht]\centering
\caption{Effect of Additional Tax on Tax Revenue} 
\label{tab:1.8}
\resizebox{\columnwidth}{!}{%
\begin{tabular}{c|cccccc|cccccc} \hline \hline
       & \multicolumn{6}{c|}{Foldable Menu}          & \multicolumn{6}{c}{Standard Logit} \\ \hline 
       & Tier I & Tier II & Tier III & Tier IV & Tier V & All & Tier I  & Tier II  & Tier III & Tier IV & Tier V & All \\ \hline 
5\% & -1.96  & 0.98  & 7.00  & 13.87  & 23.68  & 3.44 & 2.18  & 3.73 & 8.37  & 9.56   & 10.18  & 5.43 \\
10\% & -4.36  & 1.48 & 13.12   & 27.62 & 50.00   & 6.33 & 3.74  & 7.06  & 16.19 & 18.82  & 20.26 & 10.33 \\
15\% & -7.07  & 1.55  & 18.41 & 41.11 & 78.87  & 8.75 & 4.76   & 9.99   & 23.48 & 27.79  & 30.24 & 14.73 \\
20\% & -10.01  & 1.26  & 22.94  & 54.25 & 110.25  & 10.76 & 5.32   & 12.58   & 30.27  & 36.48 & 40.10 & 18.69 
\\ \hline \hline
\multicolumn{13}{p{1.2\textwidth}}{\footnotesize Notes: Each row represents the effects of a tax increase on the percentage change in tax revenue from each tier, as well as the percentage change in tax revenue from all tiers. The left panel includes the values estimated using the foldable menu model, while the right panel includes the values estimated using the standard logit model.}\\
\end{tabular}%
}
\end{table}

\subsubsection{Market Changes Under Full Availability}

This paper focuses on demand estimation using the foldable menu model, wherein consumers encounter diverse product assortments. As a natural extension, we conduct counterfactual experiments to explore how market dynamics change when all consumers have equal access to every cigarette tier. Utilizing estimated model parameters and simulated consumer incomes, we calculate anticipated sales for each cigarette tier and wholesale profits in this modified market environment. These computed values are then compared with their observed counterparts to quantify the resultant changes, as outlined in Table \ref{tab:1.9}.

\begin{table}[ht]\centering
\caption{Market Changes (\%) If All Tiers Are Available}
\resizebox{.5\columnwidth}{!}{%
\label{tab:1.9}
\begin{tabular}{cc} \hline \hline
Variables                     & Percent Change  \\ \hline
Tier I sales            & -26.90           \\
Tier II sales          & -21.37 \\
Tier III sales          & -13.49            \\
Tier IV sales          & 49.14 \\
Tier V sales           & 685.66 \\
Total Cigarette sales   & 37.81 \\
Wholesale profit              & -13.71            \\ \hline \hline
\end{tabular}
}
\end{table} 

The findings align with our expectations. When all cigarette tiers are made equally accessible to every consumer, we observe a decrease in cigarette sales of Tiers I to III, coupled with an increase in cigarette sales of Tiers IV and V. Consequently, total cigarette sales witness a notable increase of 37.81\%, indicating that rendering all tiers accessible to all consumers may not align with tobacco control objectives. Projected data suggests a remarkable sixfold increase in Tier V cigarette sales. This substantial growth can be attributed to the rise in the availability of Tier V cigarettes, estimated to be 24.36\% on average (ranging from as high as 50.43\% in 2011 to as low as 2.46\% in 2016) and projected to reach 100\% assuming full availability. The increased accessibility of low-priced cigarettes also appears to initiate smoking among non-smokers, evident in the overall surge in cigarette sales.

Our analysis reveals that cigarette wholesale profits could potentially decrease by around 13.71\% if CNTC deviated from the foldable menu approach. This observation offers valuable insights into CNTC's strategic decisions to limit the availability of low-priced cigarettes, reflecting a deliberate element of their monopoly strategy during the period.

\section{Extensions}

In this section, we relax the assumption of a monopolistic setting where the firm has complete knowledge of all consumer heterogeneities. We extend the foldable menu result to incorporate competition from other firms and random coefficients. Our analysis demonstrates that the result remains robust under these extended conditions.

\subsection{Foldable Menu with Competition}

We first demonstrate that Theorem \ref{theo1} holds even in the presence of competition from other firms. Specifically, regardless of the assortments offered by competitors, a firm can still optimize its assortment selection exclusively within the foldable menu. The proofs, detailed in Appendix \ref{Lemmas with Competition}, closely parallel those for the monopoly case. Next, we establish the existence of an equilibrium in assortment decisions. Additionally, we identify the Pareto-dominant equilibrium by beginning with a scenario in which each firm offers its single most profitable product. Firms then iteratively adjust their assortments in response to one another until the Pareto-dominant equilibrium is reached.

Consider an environment with $N$ firms engaged in assortment competition. Each firm $n$ owns a set of products $\mathcal{J}_n$. Since the best response is always a foldable menu, in any Nash equilibrium, firm $n$ selects its optimal assortment from the set $\{\vec{1}_n, \ldots, \vec{J}_n\}$, where $\vec{j}_n$ represents the top $j_n$ most profitable products for firm $n$. Let $\mathcal{L} = \{\vec{1}_1, \ldots, \vec{J}_1\} \times \ldots \times \{\vec{1}_N, \ldots, \vec{J}_N\}$ denote the set of all assortments consisting of foldable menus. We define a partial order on \(\mathcal{L}\) as:
\[
\ell \leq \ell' \quad \iff \quad \ell_n \subseteq \ell_n' \,\, \text{for all } n \in \{1, \ldots, N\}.
\]
We can show that $(\mathcal{L}, \geq)$ forms a complete lattice (see Appendix \ref{Complete Lattice}). Furthermore, within the space of assortments, an iteration operator $T$ can be defined as a best response mapping
\[
T_n(\ell) = \max_{\vec{k}_n} \quad
E\Pi_n(\vec{k}_n, \vec{j}_{-n})
\]
where $\ell = (\vec{j}_n,\vec{j}_{-n})$. We can show that the best response is non-decreasing; that is, if $\ell \leq \ell^\prime$, then $T(\ell) \leq T(\ell^\prime)$. 


\begin{lemma}
   $T$ is monotonic.
\end{lemma}

\begin{proof}
    We consider the best response of firm $n$ that maximizes its profit 
    \[
    T_n(\ell) = 
    \arg \max_{\vec{k}_n} \quad
    \sum_{j\in \vec{k}_n} \pi_{j} \frac{\exp(\delta_{j})}{1+ \sum_{j\in \vec{k}_n} \exp(\delta_{j}) + \sum_{j \in \vec{j}_{-n}} \exp(\delta_{j}) }
    \]
    $\ell \leq \ell^\prime$ implies that $\sum_{j \in \vec{j}_{-n}} \exp(\delta_{j}) \leq \sum_{j \in \vec{j}_{-n}^\prime} \exp(\delta_{j})$.
    Let $T_n(\ell) = \vec{k}_n$ be the best response to $\vec{j}_{-n}$, we know that:
    \[
    E\Pi_n(\vec{k}_n, \vec{j}_{-n}) \geq E\Pi_n(\vec{k}_{n}^\prime, \vec{j}_{-n})
    \]
    for any $\vec{k}_{n}^\prime \leq \vec{k}_{n}$. To prove the non-decreasing property of the best response, we need to show:
    \[
    E\Pi_n(\vec{k}_n, \vec{j}_{-n}^\prime) \geq E\Pi_n(\vec{k}_{n}^\prime, \vec{j}_{-n}^\prime),
    \]
    for any $\vec{j}_{-n}^\prime \geq \vec{j}_{-n}$. We can rewrite $E\Pi_n(\vec{k}_n,\vec{j}_{-n}^\prime)$ and $E\Pi_n(\vec{k}_{n}^\prime, \vec{j}_{-n}^\prime)$ as
    \begin{align*}
    E\Pi_n(\vec{k}_n, \vec{j}_{-n}^\prime) &= E\Pi_n(\vec{k}_n, \vec{j}_{-n}) \times \frac{1+ \sum_{j\in \vec{k}_n} \exp(\delta_{j}) + \sum_{j \in \vec{j}_{-n}} \exp(\delta_{j}) }{1+ \sum_{j\in \vec{k}_n} \exp(\delta_{j}) + \sum_{j \in \vec{j}_{-n}^\prime} \exp(\delta_{j}) }, \\
    E\Pi_n(\vec{k}_{n}^\prime, \vec{j}_{-n}^\prime) &= E\Pi_n(\vec{k}_{n}^\prime, \vec{j}_{-n}) \times \frac{1+ \sum_{j\in \vec{k}_{n}^\prime} \exp(\delta_{j}) + \sum_{j \in \vec{j}_{-n}} \exp(\delta_{j}) }{1+ \sum_{j\in \vec{k}_{n}^\prime} \exp(\delta_{j}) + \sum_{j \in \vec{j}_{-n}^\prime} \exp(\delta_{j}) }.
    \end{align*}
    It is clear that:
    \begin{align*}
    & \frac{1+ \sum_{j\in \vec{k}_n} \exp(\delta_{j}) + \sum_{j \in \vec{j}_{-n}} \exp(\delta_{j}) }{1+ \sum_{j\in \vec{k}_n} \exp(\delta_{j}) + \sum_{j \in \vec{j}_{-n}^\prime} \exp(\delta_{j}) }  = 1 -  \frac{\sum_{j \in \vec{j}_{-n}^\prime \backslash \vec{j}_{-n}} \exp(\delta_{j}) }{1+ \sum_{j\in \vec{k}_n} \exp(\delta_{j}) + \sum_{j \in \vec{j}_{-n}^\prime} \exp(\delta_{j}) } \\
    \geq & 1 - \frac{\sum_{j \in \vec{j}_{-n}^\prime \backslash \vec{j}_{-n}} \exp(\delta_{j}) }{1+ \sum_{j\in \vec{k}_{n}^\prime} \exp(\delta_{j}) + \sum_{j \in \vec{j}_{-n}^\prime} \exp(\delta_{j}) } =  \frac{1+ \sum_{j\in \vec{k}_{n}^\prime} \exp(\delta_{j}) + \sum_{j \in \vec{j}_{-n}} \exp(\delta_{j}) }{1+ \sum_{j\in \vec{k}_{n}^\prime} \exp(\delta_{j}) + \sum_{j \in \vec{j}_{-n}^\prime} \exp(\delta_{j}) },
    \end{align*}
    for any $\vec{k}_{n}^\prime \leq \vec{k}_{n}$. Which implies:
     \[
    E\Pi_n(\vec{k}_n, \vec{j}_{-n}^\prime) \geq E\Pi_n(\vec{k}_{n}^\prime, \vec{j}_{-n}^\prime),
    \]
    for any $\vec{j}_{-n}^\prime \geq \vec{j}_{-n}$. Thus, the best response $T$ is non-decreasing.  
\end{proof}

The monotonicity of \( T \) implies that assortments are strategic complements. When competitors add products, some consumers switch to those competing products, which increases their sales. This, in turn, raises the profit-maximizing assortment size for the focal firm.

\begin{theo}
\label{LogitOligopoly}
Consider assortment competition under logit demand, at least one equilibrium exists; in any equilibrium, firms adopt foldable menus; the set of equilibria forms a complete lattice.
\end{theo}

\begin{proof}
Given that \( T \) is non-decreasing, Topkis's fixed point theorem \citep{topkis1998supermodularity} guarantees the existence of at least one equilibrium, and the set of equilibria forms a complete lattice. Since firms only choose optimal assortments from the foldable menu, in any equilibrium, firms adopt foldable menus.
\end{proof}

Furthermore, it is straightforward to show that a firm's optimal profit decreases as its competitors expand their product offerings. Consequently, the "best" equilibrium is characterized by the fewest products among all possible equilibria. This property, combined with the non-decreasing nature of \(T\), yields the following result: starting from an initial state where each firm \(n\) offers only its most profitable product (i.e., \(\ell_n = \vec{1}_n\) for all \(n \in \{1, \ldots, N\}\)), firms iteratively adjust their assortments in response to one another's decisions. This iterative process ultimately converges to the Pareto-dominant equilibrium, which minimizes the number of products offered by each firm among all potential equilibria.

\subsection{Foldable Menu with Random Coefficients}
In this part, we relax the assumption that all consumer heterogeneities are observable to the firm by introducing random coefficients. Specifically, we assume there is a group of consumers buying from each point of sale rather than using a representative consumer. 

\begin{lemma}
\label{lemma 6}
In an assortment competition model under random coefficients logit demand, each firm includes at least its most profitable product in its assortment in any equilibrium.
\end{lemma}

\begin{proof}
    Refer to the proof for Lemma 1 in Appendix \ref{Lemmas with Competition}, in an assortment competition model under random coefficients logit demand, we have 
    \begin{align*}
    E \Pi_{\mathcal{J} \cup \{m\}} - E \Pi_{\mathcal{J}} = \int (\pi_m - E \Pi_{i,\mathcal{J}}) \frac{\exp(\delta_{im})}{1+\sum_{j \in \mathcal{J} \cup \{m\}  \cup \mathcal{J}_c} \exp(\delta_{ij})} dP_0(\delta),
    \end{align*}
    where $P_0(\delta)$ denotes the population density of $\delta_i$. Since $\pi_1 \geq E \Pi_{i,\mathcal{J}}$ for any consumer $i$ and assortment $\mathcal{J}$, it follows that $E \Pi_{\mathcal{J} \cup \{1\}} \geq E \Pi_{\mathcal{J}}$ for any $\mathcal{J}$. Therefore, in an assortment competition model under random coefficients logit demand, each firm includes at least its most profitable product in its assortment in any equilibrium.
\end{proof}

\begin{theo}
\label{RandomOligopoly}
In an assortment competition model under logit demand with random coefficients, if firms restrict their strategies to foldable menus and $T$ is monotonic, at least one equilibrium exists and the set of equilibria forms a complete lattice.
\end{theo}

The empirical literature often employs a separable random coefficient model defined as $\delta_{ij} = \delta_j + \eta_{ij}$, where $\eta_{i\cdot}$ is a mean-zero random vector that is i.i.d. across consumers, and firms observe only $\delta_j$. Under this framework, firm $n$ achieves the following expected profit from a foldable assortment $(\vec{k}_n, \vec{j}_{-n})$:
\[
E\Pi_n(\vec{k}_n,\vec{j}_{-n}) = \sum_{j\in \vec{k}_n} \pi_{j} \exp(\delta_{j}) \int \frac{\exp(\eta_{ij})}{1+ \sum_{j\in \vec{k}_n} \exp(\delta_{j}+\eta_{ij}) + \sum_{j \in \vec{j}_{-n}} \exp(\delta_{j}+\eta_{ij}) } d\eta_{i\cdot}
\]
In this setting, a sufficient condition for the monotonicity of $T$ is:
    \[
    \frac{ \int \frac{\exp(\eta_{ij})}{1+ \sum_{j\in \vec{k}_n} \exp(\delta_{j}+\eta_{ij}) + \sum_{j \in \vec{j}_{-n}^\prime} \exp(\delta_{j}+\eta_{ij}) } d\eta_{i\cdot}}
    { \int \frac{\exp(\eta_{ij})}{1+ \sum_{j\in \vec{k}_n} \exp(\delta_{j}+\eta_{ij}) + \sum_{j \in \vec{j}_{-n}} \exp(\delta_{j}+\eta_{ij}) } d\eta_{i\cdot}} >
    \frac{ \int \frac{\exp(\eta_{ij})}{1+ \sum_{j\in \vec{k}_n^\prime} \exp(\delta_{j}+\eta_{ij}) + \sum_{j \in \vec{j}_{-n}^\prime} \exp(\delta_{j}+\eta_{ij}) } d\eta_{i\cdot}}
    { \int \frac{\exp(\eta_{ij})}{1+ \sum_{j\in \vec{k}_n^\prime} \exp(\delta_{j}+\eta_{ij}) + \sum_{j \in \vec{j}_{-n}} \exp(\delta_{j}+\eta_{ij}) } d\eta_{i\cdot}},
    \] 
for any $\vec{k}_n^\prime \leq \vec{k}_n$ and $\vec{j}_{-n}^\prime \geq \vec{j}_{-n}$.

Foldable menus are particularly appealing due to their simplicity in implementation. Additionally, the iterative approach described in the previous subsection can be employed to identify the ``best'' Nash equilibrium. The optimality of restricting to foldable menus can be readily verified, and any potential loss from this restriction can be calculated. Specifically, given the competing assortments in a foldable menu equilibrium, firm $n$ need only compare its foldable menu to its remaining $2^{J_n} - 1$ feasible assortments. This comparison is computationally efficient, as firms generally own a small number of products in most practical applications. Appendix \ref{SimuRandCoeff} confirms that, under our estimated model, foldable menus closely approximate the optimal menus across a wide range of consumer heterogeneity.



\section{Conclusion}

This paper presents a novel approach to demand estimation that accounts for unobserved choice set heterogeneity while relying solely on aggregate sales data. In the absence of price discrimination, a monopoly firm serving each market can maximize expected profit by selecting optimal assortments from a "foldable menu," where the number of assortments corresponds to the number of available products. We recover unobserved choice sets using the firm's profit maximization conditions and propose an estimation procedure for the model.

We apply our method to real data from China’s tobacco industry to demonstrate its effectiveness. The results highlight the critical importance of accounting for choice set heterogeneity, as neglecting it can lead to underestimated price elasticities. This underestimation stems from misattributing limited product substitution to inelastic demand. Consequently, tax policies based on a misspecified model may be misleading and produce unintended consequences. Simulated tax increases reveal that failing to account for endogenous assortments results in underestimated declines in higher-tier product sales, incorrect directional predictions for lower-tier sales, and an overestimation of tax revenue by more than 50\%. Additionally, making all cigarette tiers available leads to an increase in total cigarette sales (+37.81\%) but reduces wholesale profits (-13.71\%).

Our method is flexible enough to account for richer consumer heterogeneity. In our application, we specifically address heterogeneity in price sensitivity by calculating a set of price sensitivity cutoffs for each market during the estimation of our empirical model. When consumer heterogeneity extends across multiple dimensions, these cutoffs become contingent on consumer-specific quality measures, which introduces computational challenges. However, aside from this complexity, the estimation process remains largely unchanged.

Finally, we extend our methodology to settings that include competition and random coefficients. We demonstrate that the foldable menu result remains robust in these scenarios, proving the existence of at least one equilibrium in assortment competition, which can be identified using an iterative procedure. These results underscore the flexibility and wide applicability of our approach across diverse contexts.

\newpage

\baselineskip14pt

\begin{appendices}

\setstretch{1.5}

\section{Proof of Lemmas}

\subsection{Lemma 1}
\label{Lemma 1}

\begin{proof}
\emph{(Lemma 1)}
Under the logit model, we demonstrate that, for any consumer $i$, the inequality $\pi_m \geq E \Pi_{i,\mathcal{J}}$ is equivalent to $E \Pi_{i,\mathcal{J} \cup {m}} \geq E \Pi_{i,\mathcal{J}}$. 
\begin{align}
E \Pi_{i,\mathcal{J} \cup \{m\}} = &
\sum_{j \in \mathcal{J} \cup \{m\}} \pi_j \frac{\exp(\delta_{ij})}{1+\sum_{j \in \mathcal{J} \cup \{m\}} \exp(\delta_{ij})} \nonumber \\
= & \sum_{j \in \mathcal{J}} \pi_j \frac{\exp(\delta_{ij})}{1+\sum_{j \in \mathcal{J} \cup \{m\}} \exp(\delta_{ij})} + \pi_m \frac{\exp(\delta_{im})}{1+\sum_{j \in \mathcal{J} \cup \{m\}} \exp(\delta_{ij})} \nonumber \\
= & E \Pi_{i,\mathcal{J}} \frac{1+\sum_{j \in \mathcal{J}} \exp(\delta_{ij})}{1+\sum_{j \in \mathcal{J} \cup \{m\}} \exp(\delta_{ij})} + \pi_m \frac{\exp(\delta_{im})}{1+\sum_{j \in \mathcal{J} \cup \{m\}} \exp(\delta_{ij})}.
\end{align}
The last equality follows from the definition of $E \Pi_{i,\mathcal{J}}$, where
\[
 E \Pi_{i,\mathcal{J}}  = \sum_{j \in \mathcal{J}} \pi_j \frac{\exp(\delta_{ij})}{1+\sum_{j \in \mathcal{J}} \exp(\delta_{ij})}.
\]
Then, we have
\begin{align}
E \Pi_{i,\mathcal{J} \cup \{m\}} - E \Pi_{i,\mathcal{J}} = & E \Pi_{i,\mathcal{J}} \frac{1+\sum_{j \in \mathcal{J}} \exp(\delta_{ij})}{1+\sum_{j \in \mathcal{J} \cup \{m\}} \exp(\delta_{ij})} + \pi_m \frac{\exp(\delta_{im})}{1+\sum_{j \in \mathcal{J} \cup \{m\}} \exp(\delta_{ij})} - E \Pi_{i,\mathcal{J}} \nonumber \\
= & (\pi_m - E \Pi_{i,\mathcal{J}}) \frac{\exp(\delta_{im})}{1+\sum_{j \in \mathcal{J} \cup \{m\}} \exp(\delta_{ij})},
\end{align}
so that $\pi_m \geq E \Pi_{i,\mathcal{J}}$ is equivalent to $E \Pi_{i,\mathcal{J} \cup {m}} \geq E \Pi_{i,\mathcal{J}}$. Therefore, it is more profitable to add a new option with a higher profit-margin than the current average profit. Given that we arrange the products in descending order of unit margin, $\pi_1 \geq E \Pi_{i,\mathcal{J}}$ is trivially satisfied, because $\pi_1 \geq \pi_j$ for all $j \in \mathcal{J}$, and $E \Pi_{i,\mathcal{J}}$ represents a weighted average of $\pi_j$ in $\mathcal{J}$. 
\end{proof}

\subsection{Lemma 2}
\label{Lemma 2}

\begin{proof}
\emph{(Lemma 2)}
First, we show that if $E \Pi_{i,\mathcal{J}\cup\{m\}} \leq E \Pi_{i,\mathcal{J}}$, then $E \Pi_{i, \mathcal{J}\cup\{m,m^{\prime}\}}  \leq E \Pi_{i,\mathcal{J}\cup\{m\}}$ for all $m^{\prime}>m$. That is, if adding $m$ is worse, adding $m^{\prime}$ is worse on top of adding $m$. $E \Pi_{i,\mathcal{J}\cup\{m\}} \leq E \Pi_{i,\mathcal{J}}$ implies that $\pi_{m} \leq  E \Pi_{i,\mathcal{J}}$. We want to show that $\pi_{m^{\prime}} \leq E \Pi_{i,\mathcal{J}\cup\{m\}}$.
\begin{align}
E \Pi_{i,\mathcal{J}\cup\{m\}}  = & \sum_{j\in\mathcal{J}\cup\{m\}}\pi_{j}\frac{\exp(\delta_{ij})}{1+\sum_{j\in\mathcal{J}\cup\{m\}}\exp(\delta_{ij})}	\nonumber \\
= & \sum_{j\in\mathcal{J}}\pi_{j}\frac{\exp(\delta_{ij})}{1+\sum_{j\in\mathcal{J}\cup\{m\}}\exp(\delta_{ij})}+\pi_{m}\frac{\exp(\delta_{im})}{1+\sum_{j\in\mathcal{J}\cup\{m\}}\exp(\delta_{ij})} \nonumber \\
= & E \Pi_{i,\mathcal{J}} \frac{1+\sum_{j \in \mathcal{J}} \exp(\delta_{ij})}{1+\sum_{j \in \mathcal{J} \cup \{m\}} \exp(\delta_{ij})} +\pi_{m}\frac{\exp(\delta_{im})}{1+\sum_{j\in\mathcal{J}\cup\{m\}}\exp(\delta_{ij})} \nonumber \\
\geq & \pi_m \frac{1+\sum_{j\in\mathcal{J}}\exp(\delta_{ij})}{1+\sum_{j\in\mathcal{J}\cup\{m\}}\exp(\delta_{ij})}+\pi_{m}\frac{\exp(\delta_{im})}{1+\sum_{j\in\mathcal{J}\cup\{m\}}\exp(\delta_{ij})}  \nonumber \\
= & \pi_{m}\geq\pi_{m^{\prime}}.
\end{align}
The third equality follows from the definition of $E \Pi_{i,\mathcal{J}}$, and the first inequality follows from $\pi_m \leq E \Pi_{i,\mathcal{J}}$.

Second, we show that if $E \Pi_{i,\mathcal{J}\cup\{m,m^{\prime}\}} \geq E \Pi_{i,\mathcal{J}\cup\{m\}}$, then $E \Pi_{i, \mathcal{J}\cup\{m\}}  \geq E \Pi_{i,\mathcal{J}}$ for all $m < m^{\prime}$. That is, if removing $m^{\prime}$ is worse, removing $m$ is worse on top of removing $m^{\prime}$. $E \Pi_{i,\mathcal{J}\cup\{m,m^{\prime}\}} \geq E \Pi_{i,\mathcal{J}\cup\{m\}}$ implies that $\pi_{m^{\prime}} \geq  E \Pi_{i,\mathcal{J}\cup\{m\}}$. We want to show that $\pi_{m} \geq E \Pi_{i,\mathcal{J}}$. We have
\begin{align}
E \Pi_{i,\mathcal{J}\cup\{m\}}  = & \sum_{j\in\mathcal{J}\cup\{m\}}\pi_{j}\frac{\exp(\delta_{ij})}{1+\sum_{j\in\mathcal{J}\cup\{m\}}\exp(\delta_{ij})}	\nonumber \\
= & \sum_{j\in\mathcal{J}}\pi_{j}\frac{\exp(\delta_{ij})}{1+\sum_{j\in\mathcal{J}\cup\{m\}}\exp(\delta_{ij})}+\pi_{m}\frac{\exp(\delta_{im})}{1+\sum_{j\in\mathcal{J}\cup\{m\}}\exp(\delta_{ij})} \nonumber \\
= & E \Pi_{i,\mathcal{J}} \frac{1+\sum_{j \in \mathcal{J}} \exp(\delta_{ij})}{1+\sum_{j \in \mathcal{J} \cup \{m\}} \exp(\delta_{ij})} +\pi_{m}\frac{\exp(\delta_{im})}{1+\sum_{j\in\mathcal{J}\cup\{m\}}\exp(\delta_{ij})} \nonumber \\
\leq & \pi_{m^{\prime}} \leq \pi_{m}.
\end{align}
The third equality follows from the definition of $E \Pi_{i,\mathcal{J}}$. After moving
\[
\pi_{m}\frac{\exp(\delta_{im})}{1+\sum_{j\in\mathcal{J}\cup\{m\}}\exp(\delta_{ij})}
\]
to the RHS of the last inequality, we obtain
\[
E \Pi_{i,\mathcal{J}} \frac{1+\sum_{j \in \mathcal{J}} \exp(\delta_{ij})}{1+\sum_{j \in \mathcal{J} \cup \{m\}} \exp(\delta_{ij})} \leq \pi_m \frac{1+\sum_{j \in \mathcal{J}} \exp(\delta_{ij})}{1+\sum_{j \in \mathcal{J} \cup \{m\}} \exp(\delta_{ij})},
\]
which implies $\pi_{m} \geq E \Pi_{i,\mathcal{J}}$.
\end{proof}

\subsection{Lemma 3}
\label{Lemma 3}

\begin{proof}
\emph{(Lemma 3)}
Given Lemma 2, when $\vec{m}$ is optimal among $\{\vec{1},\vec{2},\cdots,\vec{J}\}$, we have $\pi_m \geq E \Pi_{i,\vec{j}}$ for any $j < m$, and $E \Pi_{i,\vec{m}} > \pi_j$ for any $j > m$. Now, we show that $\vec{m}$ also dominates all non-consecutive assortments, such as $\{1,3\}$, $\{1,2,4\}$, etc.

Start with the case $m \geq 3$. First, we compare $\vec{m}$ and its own subsets. Let $\mathbf{k}$ be the set of missing internal products of $\vec{m}$, i.e., $\forall k \in \mathbf{k}, 1<k<m$. Let $\vec{m} \backslash \mathbf{k}$ be the non-consecutive assortment, we have
\begin{align}
& E \Pi_{i,\vec{m}} - E \Pi_{i,\vec{m} \backslash \mathbf{k}} \nonumber \\
= &  \frac{\sum_{j \in \vec{m}} \pi_j \exp(\delta_{ij})}{1 + \sum_{j \in \vec{m}} \exp(\delta_{ij})} - \frac{\sum_{j \in \vec{m} \backslash \mathbf{k}}\pi_j \exp(\delta_{ij})}{1 + \sum_{j \in \vec{m} \backslash \mathbf{k}} \exp(\delta_{ij})} \nonumber \\
= & \frac{(\sum_{j \in \vec{m}} \pi_j \exp(\delta_{ij}))(1 + \sum_{j \in \vec{m} \backslash \mathbf{k}} \exp(\delta_{ij})) - (\sum_{j \in \vec{m} \backslash \mathbf{k}}\pi_j \exp(\delta_{ij}))(1 + \sum_{j \in \vec{m}} \exp(\delta_{ij}))} {(1 + \sum_{j \in \vec{m}} \exp(\delta_{ij})) (1 + \sum_{j \in \vec{m} \backslash \mathbf{k}} \exp(\delta_{ij}))}.
\end{align}
We can simplify the numerator as
\begin{align}
& \sum_{j \in \vec{m}} \pi_j \exp(\delta_{ij}) + [\sum_{j \in \vec{m}} \pi_j \exp(\delta_{ij})][\sum_{j \in \vec{m}} \exp(\delta_{ij}) - \sum_{j \in  \mathbf{k}} \exp(\delta_{ij})] \nonumber \\
- & \sum_{j \in \vec{m}\backslash \mathbf{k}} \pi_j \exp(\delta_{ij}) - [\sum_{j \in \vec{m}} \pi_j \exp(\delta_{ij}) - \sum_{j \in  \mathbf{k}} \pi_j \exp(\delta_{ij})][\sum_{j \in \vec{m}} \exp(\delta_{ij})] \nonumber \\
= & \sum_{j \in \mathbf{k}} \pi_j \exp(\delta_{ij}) - [\sum_{j \in \vec{m}} \pi_j \exp(\delta_{ij})][\sum_{j \in  \mathbf{k}} \exp(\delta_{ij})] + [\sum_{j \in  \mathbf{k}} \pi_j \exp(\delta_{ij})][\sum_{j \in \vec{m}} \exp(\delta_{ij})].
\end{align}
For each $k \in  \mathbf{k}$, we have
\begin{align}
\label{eq:16}
& \pi_k \exp(\delta_{ik}) - [\sum_{j \in \vec{m}} \pi_j \exp(\delta_{ij})]\exp(\delta_{ik}) +  \pi_k \exp(\delta_{ik})[\sum_{j \in \vec{m}} \exp(\delta_{ij})] \nonumber \\
= & [\pi_k (1+\sum_{j \in \vec{m}} \exp(\delta_{ij})) - \sum_{j \in \vec{m}} \pi_j \exp(\delta_{ij})]\exp(\delta_{ik})\nonumber \\
= & [\pi_k (1+\sum_{j \in \vec{k}} \exp(\delta_{ij})) - \sum_{j \in \vec{k}} \pi_j \exp(\delta_{ij}) + \sum_{j=k+1}^m(\pi_k - \pi_j) \exp(\delta_{ij})]\exp(\delta_{ik})\nonumber \\
\geq & [\pi_k (1+\sum_{j \in \vec{k}} \exp(\delta_{ij})) - \sum_{j \in \vec{k}} \pi_j \exp(\delta_{ij})]\exp(\delta_{ik})\nonumber \\
\geq & [\pi_m (1+\sum_{j \in \vec{k}} \exp(\delta_{ij})) - \sum_{j \in \vec{k}} \pi_j \exp(\delta_{ij})]\exp(\delta_{ik})\nonumber \\
= & (\pi_m - E \Pi_{i,\vec{k}})(1+\sum_{j \in \vec{k}} \exp(\delta_{ij}))\exp(\delta_{ik})\geq 0.
\end{align}
The second equality follows from  isolating the non-negative term $\sum_{j=k+1}^m(\pi_k - \pi_j) \exp(\delta_{ij})$,  the second inequality follows from $\pi_m \leq \pi_k$ for any $k < m$, the last equality follows from the definition of $E \Pi_{i,\vec{k}}$, and the last inequality follows from $E \Pi_{i,\vec{k}} \leq \pi_m$ for any $k < m$. So we have $E \Pi_{i,\vec{m}} \geq E \Pi_{i,\vec{m} \backslash \mathbf{k}}$.

Second, for the case $n > m$, let $\mathbf{k}$ be the set of missing internal products of $\vec{n}$, i.e., $\forall k \in \mathbf{k}, 1<k<n$. Let $\vec{n} \backslash \mathbf{k}$ be the non-consecutive assortment, we have
\begin{align}
& E \Pi_{i,\vec{m}} - E \Pi_{i,\vec{n} \backslash \mathbf{k}} \nonumber \\
= &  \frac{\sum_{j \in \vec{m}} \pi_j \exp(\delta_{ij})}{1 + \sum_{j \in \vec{m}} \exp(\delta_{ij})} - \frac{\sum_{j \in \vec{n} \backslash \mathbf{k}}\pi_j \exp(\delta_{ij})}{1 + \sum_{j \in \vec{n} \backslash \mathbf{k}} \exp(\delta_{ij})} \nonumber \\
= & \frac{(\sum_{j \in \vec{m}} \pi_j \exp(\delta_{ij}))(1 + \sum_{j \in \vec{n} \backslash \mathbf{k}} \exp(\delta_{ij})) - (\sum_{j \in \vec{n} \backslash \mathbf{k}}\pi_j \exp(\delta_{ij}))(1 + \sum_{j \in \vec{m}} \exp(\delta_{ij}))} {(1 + \sum_{j \in \vec{m}} \exp(\delta_{ij})) (1 + \sum_{j \in \vec{n} \backslash \mathbf{k}} \exp(\delta_{ij}))}.
\end{align}
We can simplify the numerator as follows:
\begin{align}
\label{eq:18}
& \sum_{j \in \vec{m}} \pi_j \exp(\delta_{ij}) + [\sum_{j \in \vec{m}} \pi_j \exp(\delta_{ij})][\sum_{j \in \vec{n}} \exp(\delta_{ij}) - \sum_{j \in  \mathbf{k}} \exp(\delta_{ij})] \nonumber \\
- & \sum_{j \in \vec{n}\backslash \mathbf{k}} \pi_j \exp(\delta_{ij}) - [\sum_{j \in \vec{n}} \pi_j \exp(\delta_{ij}) - \sum_{j \in  \mathbf{k}} \pi_j \exp(\delta_{ij})][\sum_{j \in \vec{m}} \exp(\delta_{ij})] \nonumber \\
= & \sum_{j \in \mathbf{k}} \pi_j \exp(\delta_{ij}) - [\sum_{j \in \vec{m}} \pi_j \exp(\delta_{ij})][\sum_{j \in  \mathbf{k}} \exp(\delta_{ij})] + [\sum_{j \in  \mathbf{k}} \pi_j \exp(\delta_{ij})][\sum_{j \in \vec{m}} \exp(\delta_{ij})] \nonumber \\
- & \sum_{r = m+1}^n\pi_r \exp(\delta_{ir}) + [\sum_{j \in \vec{m}} \pi_j \exp(\delta_{ij})][\sum_{j \in \vec{n}} \exp(\delta_{ij})] - [\sum_{j \in \vec{n}} \pi_j \exp(\delta_{ij})][\sum_{j \in \vec{m}} \exp(\delta_{ij})].
\end{align}
We can further simplify the second line of Equation \eqref{eq:18} as follows:
\begin{align}
& - \sum_{r = m+1}^n\pi_r \exp(\delta_{ir}) + [\sum_{j \in \vec{m}} \pi_j \exp(\delta_{ij})][\sum_{j \in \vec{n}} \exp(\delta_{ij})] - [\sum_{j \in \vec{n}} \pi_j \exp(\delta_{ij})][\sum_{j \in \vec{m}} \exp(\delta_{ij})] \nonumber \\
= & - \sum_{r = m+1}^n\pi_r \exp(\delta_{ir}) + [\sum_{j \in \vec{m}} \pi_j \exp(\delta_{ij})][\sum_{j \in \vec{m}} \exp(\delta_{ij}) + \sum_{r = m+1}^n \exp(\delta_{ir})] \nonumber \\
& - [\sum_{j \in \vec{m}} \pi_j \exp(\delta_{ij}) + \sum_{r = m+1}^n\pi_r \exp(\delta_{ir})][\sum_{j \in \vec{m}} \exp(\delta_{ij})] \nonumber \\
= & - \sum_{r = m+1}^n\pi_r \exp(\delta_{ir}) + [\sum_{j \in \vec{m}} \pi_j \exp(\delta_{ij})][\sum_{r = m+1}^n \exp(\delta_{ir})] - [\sum_{r = m+1}^n\pi_r \exp(\delta_{ir})][\sum_{j \in \vec{m}} \exp(\delta_{ij})].
\end{align}
Note for any $r \in [m+1,n]$, we have
\begin{align}
& - \pi_r \exp(\delta_{ir}) + [\sum_{j \in \vec{m}} \pi_j \exp(\delta_{ij})] \exp(\delta_{ir}) - \pi_r \exp(\delta_{ir})[\sum_{j \in \vec{m}} \exp(\delta_{ij})] \nonumber \\
= & [\sum_{j \in \vec{m}} \pi_j \exp(\delta_{ij}) - \pi_r(1+\sum_{j \in \vec{m}} \exp(\delta_{ij}))]\exp(\delta_{ir})\nonumber \\
= & (E\Pi_{i,\vec{m}} - \pi_r)(1+\sum_{j \in \vec{m}} \exp(\delta_{ij}))\exp(\delta_{ir}) > 0.
\end{align}
The last equality follows from the definition of 
$E \Pi_{i,\vec{m}}$, and the last inequality follows from $E \Pi_{i,\vec{m}} > \pi_{r}$ for any $r > m$. 

For any $k > m$, the part $\pi_k \exp(\delta_{ik}) - [\sum_{j \in \vec{m}} \pi_j \exp(\delta_{ij})] \exp(\delta_{ik}) + \pi_k \exp(\delta_{ik})[\sum_{j \in \vec{m}} \exp(\delta_{ij})]$ will cancel out with one of $-\pi_r \exp(\delta_{ir}) + [\sum_{j \in \vec{m}} \pi_j \exp(\delta_{ij})] \exp(\delta_{ir}) - \pi_r \exp(\delta_{ir})[\sum_{j \in \vec{m}} \exp(\delta_{ij})]$. For any $k \leq m$, from Equation \eqref{eq:16}, we know
$\pi_k \exp(\delta_{ik}) - [\sum_{j \in \vec{m}} \pi_j \exp(\delta_{ij})] \exp(\delta_{ik}) + \pi_k \exp(\delta_{ik})[\sum_{j \in \vec{m}} \exp(\delta_{ij})] \geq 0.$ So we have 
$E \Pi_{i,\vec{m}} \geq E \Pi_{i,\vec{n} \backslash \mathbf{k}}$.

Last, for the case $n < m$, let $\mathbf{k}$ be the set of missing internal products of $\vec{n}$, i.e., $\forall k \in \mathbf{k}, 1<k<n$. Let $\vec{n} \backslash \mathbf{k}$ be the non-consecutive assortment, we have
\begin{align}
& E \Pi_{i,\vec{m}} - E \Pi_{i,\vec{n} \backslash \mathbf{k}} \nonumber \\
= &  \frac{\sum_{j \in \vec{m}} \pi_j \exp(\delta_{ij})}{1 + \sum_{j \in \vec{m}} \exp(\delta_{ij})} - \frac{\sum_{j \in \vec{n} \backslash \mathbf{k}}\pi_j \exp(\delta_{ij})}{1 + \sum_{j \in \vec{n} \backslash \mathbf{k}} \exp(\delta_{ij})} \nonumber \\
= & \frac{(\sum_{j \in \vec{m}} \pi_j \exp(\delta_{ij}))(1 + \sum_{j \in \vec{n} \backslash \mathbf{k}} \exp(\delta_{ij})) - (\sum_{j \in \vec{n} \backslash \mathbf{k}}\pi_j \exp(\delta_{ij}))(1 + \sum_{j \in \vec{m}} \exp(\delta_{ij}))} {(1 + \sum_{j \in \vec{m}} \exp(\delta_{ij})) (1 + \sum_{j \in \vec{n} \backslash \mathbf{k}} \exp(\delta_{ij}))}.
\end{align}
We can simplify the numerator as follows:
\begin{align}
\label{eq:22}
& \sum_{j \in \vec{m}} \pi_j \exp(\delta_{ij}) + [\sum_{j \in \vec{m}} \pi_j \exp(\delta_{ij})][\sum_{j \in \vec{n}} \exp(\delta_{ij}) - \sum_{j \in  \mathbf{k}} \exp(\delta_{ij})] \nonumber \\
- & \sum_{j \in \vec{n}\backslash \mathbf{k}} \pi_j \exp(\delta_{ij}) - [\sum_{j \in \vec{n}} \pi_j \exp(\delta_{ij}) - \sum_{j \in  \mathbf{k}} \pi_j \exp(\delta_{ij})][\sum_{j \in \vec{m}} \exp(\delta_{ij})] \nonumber \\
= & \sum_{j \in \mathbf{k}} \pi_j \exp(\delta_{ij}) - [\sum_{j \in \vec{m}} \pi_j \exp(\delta_{ij})][\sum_{j \in  \mathbf{k}} \exp(\delta_{ij})] + [\sum_{j \in  \mathbf{k}} \pi_j \exp(\delta_{ij})][\sum_{j \in \vec{m}} \exp(\delta_{ij})] \nonumber \\
+ & \sum_{r = n+1}^m\pi_r \exp(\delta_{ir}) + [\sum_{j \in \vec{m}} \pi_j \exp(\delta_{ij})][\sum_{j \in \vec{n}} \exp(\delta_{ij})] - [\sum_{j \in \vec{n}} \pi_j \exp(\delta_{ij})][\sum_{j \in \vec{m}} \exp(\delta_{ij})]
\end{align}
We can further simplify the second line of Equation \eqref{eq:22} as follows:
\begin{align}
& \sum_{r = n+1}^m\pi_r \exp(\delta_{ir}) + [\sum_{j \in \vec{m}} \pi_j \exp(\delta_{ij})][\sum_{j \in \vec{n}} \exp(\delta_{ij})] - [\sum_{j \in \vec{n}} \pi_j \exp(\delta_{ij})][\sum_{j \in \vec{m}} \exp(\delta_{ij})] \nonumber \\
= & \sum_{r = n+1}^m\pi_r \exp(\delta_{ir}) + [\sum_{j \in \vec{n}} \pi_j \exp(\delta_{ij}) + \sum_{r = n+1}^m \pi_r \exp(\delta_{ir})][\sum_{j \in \vec{n}} \exp(\delta_{ij})] \nonumber \\
- & [\sum_{j \in \vec{n}} \pi_j \exp(\delta_{ij})][\sum_{j \in \vec{n}} \exp(\delta_{ij}) +  \sum_{r = n+1}^m \exp(\delta_{ir})] \nonumber \\
= & \sum_{r = n+1}^m\pi_r \exp(\delta_{ir}) + [\sum_{r = n+1}^m  \pi_r \exp(\delta_{ir})][\sum_{j \in \vec{n}} \exp(\delta_{ij})] - [\sum_{j \in \vec{n}} \pi_j \exp(\delta_{ij})][\sum_{r = n+1}^m \exp(\delta_{ir})].
\end{align}
Note for any $r \in [n+1,m]$, we have
\begin{align}
& \pi_r \exp(\delta_{ir}) +  \pi_r \exp(\delta_{ir})[\sum_{j \in \vec{n}} \exp(\delta_{ij})] - [\sum_{j \in \vec{n}} \pi_j \exp(\delta_{ij})] \exp(\delta_{ir}) \nonumber \\
= & [\pi_r(1+\sum_{j \in \vec{n}} \exp(\delta_{ij})) - \sum_{j \in \vec{n}} \pi_j \exp(\delta_{ij})]\exp(\delta_{ir})\nonumber \\
= & (\pi_r - E\Pi_{i,\vec{n}})(1+\sum_{j \in \vec{n}} \exp(\delta_{ij}))\exp(\delta_{ir})\nonumber \\
\geq & (\pi_m - E\Pi_{i,\vec{n}})(1+\sum_{j \in \vec{n}} \exp(\delta_{ij}))\exp(\delta_{ir}) \geq 0.
\end{align}
The last equality follows from the definition of 
$E \Pi_{i,\vec{n}}$, the first inequality follows from $\pi_r \geq \pi_m$ for any $r \leq m$, and the last inequality follows from $E \Pi_{i,\vec{n}} \leq \pi_m$ for any $n < m$. 

For any $k < m$, from Equation \eqref{eq:16}, we know
$\pi_k \exp(\delta_{ik}) - [\sum_{j \in \vec{m}} \pi_j \exp(\delta_{ij})] \exp(\delta_{ik}) + \pi_k \exp(\delta_{ik})[\sum_{j \in \vec{m}} \exp(\delta_{ij})] \geq 0.$ So, we have $E \Pi_{i,\vec{m}} \geq E \Pi_{i,\vec{n} \backslash \mathbf{k}}$.
\end{proof} 

\section{Summary Statistics of Cigarette Sales}
\label{Summary Statistics of Cigarette Sales}
\begin{table}[ht]\centering
\caption{Summary Statistics of Cigarette Sales}
\resizebox{.8\columnwidth}{!}{%
\label{tab:1.10}
\begin{tabular}{c|cccccc}  \hline \hline
\diagbox[width=5em]{Tier}{Year}                    & 2011   & 2012   & 2013   & 2014   & 2015   & 2016   \\ \hline
\multirow{2}{*}{I} & 99.37  & 120.26 & 138.55 & 161.8 & 168.21 & 156.07 \\
                   &  (84.1)  & (97.87)  & (111.01) & (127.89) & (136.13) & (132.12) \\ \hline
\multirow{2}{*}{II} &  48.02  & 60.19  & 69.88  & 83.53  & 91.88  & 95.71  \\ 
                   &  (43.14)  & (50.98)  & (57.84)  & (66.89)  & (67.72)  & (70.56)  \\ \hline
\multirow{2}{*}{III} &  297.44 & 356.2 & 369.41 & 362.23 & 347.95 & 326.51 \\
                   & (216.02) & (239.92) & (248.82) & (236.06) & (228.35) & (219.95) \\ \hline
\multirow{2}{*}{IV} & 219.48 & 169.19 & 155.18 & 149 & 135.76 & 125.93 \\
                   &  (138.67) & (100.73) & (92.32)  & (91.59)  & (87.19 ) & (84.19)  \\ \hline
\multirow{2}{*}{V} &  102.04  & 76.81  & 58.64  & 51.25  & 45.82  & 40.9  \\
                   &  (74.06)  & (53.88)  & (40.68)  & (38.97)  & (36.78)  & (33.56) \\ \hline
N & 31 & 31 & 31 & 31 & 31 & 31 \\ \hline \hline
\multicolumn{7}{p{0.8\textwidth}}{\footnotesize Notes: This table is based on the information provided by the China Tobacco Yearbooks. Data are in units of 100 million cigarettes. Values are the mean across 31 provinces. Standard deviations are in parentheses.}\\
\end{tabular}
}
\end{table}

\section{Cigarette Tax Increases}
\label{Cigarette Tax Increases}
Two exogenous tax increases in the industry occurred in May 2009 and May 2015. Table \ref{tab:1.11} offers a summary of the details of these tax increases. The allocation price mentioned in the table is the price negotiated between STMA and the tax authority in China, specifically the State Administration of Taxation. This allocation price serves as the tax base for cigarette ad valorem excise tax, and STMA utilizes it as a criterion to categorize cigarettes into five tiers \citep{gao2012can}.

The tax adjustment in 2009 primarily entailed: (i) excise tax rates for Grade A (Tier I and II) cigarettes increasing from 45\% to 56\%, and for Grade B (Tier III, IV, and V) cigarettes increasing from 30\% to 36\%; (ii) an additional 5\% ad valorem excise tax applied at the wholesale level; (iii) adjusted standards for Grade A and B cigarettes, where products with allocation prices greater than or equal to \yen 7 per pack (of 20 cigarettes) were considered Grade A, while those with allocation prices less than \yen 7 were considered Grade B. In contrast, the price threshold was \yen 5 before the tax adjustment. Accordingly, the cigarette classification standards for Tiers II and III also changed. The Chinese government raised cigarette excise tax again at the wholesale level in 2015, increasing the ad valorem excise tax rate from 5\% to 11\% and adding a \yen 0.10/pack specific tax for all tiers \citep{zheng2018tobacco}.

\begin{table}[ht]\centering
\caption{China's Cigarette Excise Tax Adjustments}
\label{tab:1.11}
\resizebox{\columnwidth}{!}{%
\begin{tabular}{llll} \\ \hline \hline
                           & Before May 2009                     & After May 2009     & After May 2015                  \\ \hline
\textbf{Producer level} & & \\
Specific excise tax  & \yen 0.06  /pack                             & \yen 0.06 /pack    & \yen 0.06  /pack                          \\
Ad valorem excise tax rate   &                                     &                                     \\
Class A cigarettes           & Allocation price $\geq$ \yen 5  /pack & Allocation price $\geq$ \yen 7  /pack & Allocation price $\geq$ \yen 7  /pack \\
                             & 45\%                                & 56\%          & 56\%                        \\
Class B cigarettes           &  Allocation price \textless{} \yen 5  /pack     & Allocation price \textless{} \yen 7  /pack    & Allocation price \textless{} \yen 7  /pack  \\
                             & 30\%                                & 36\%          & 36\%                      \\
\textbf{Wholesale level} & & \\
Specific excise tax  & \yen  0  /pack                             & \yen 0  /pack    & \yen 0.10  /pack                          \\
Ad valorem excise tax rate          & 0\%                                 & 5\%   & 11\%  \\
\hline\hline  
\multicolumn{4}{p{1.3\textwidth}}{\footnotesize Notes: Grade A cigarettes encompass Tiers I and II. Grade B cigarettes encompass Tiers III, IV and V. The standard changed in May 2009.}\\                            
\end{tabular}
}
\end{table}

\section{Wholesale Margin Calculations}
\label{Wholesale Profits}

Equation \eqref{eq:pricing} (first introduced in \cite{gao2012can}) demonstrates China's cigarette pricing mechanism where $A$ represents allocation price, $a$ indicates the allocation-wholesale profit margin, $b$ indicates the wholesale-retail profit margin, and $Rtvat$ indicates the VAT (value-added tax) rate. $P_r$ is the retail price at which retailers sell cigarettes to consumers, and it is the final price of tobacco products:
\begin{equation}
\label{eq:pricing}
P_r = A \times (1+a) \times (1+b) \times (1+Rtvat).
\end{equation}
As addressed, retail price is equal to the allocation price plus the allocation-wholesale margin, wholesale-retail margin, and VAT tax. $A$ includes excise tax but excludes VAT, so no excise tax appears in equation \eqref{eq:pricing}. VAT is collected at all circulation segments, including produce, wholesale, and retail, and uses the added value at each level as the tax base. Therefore, placing the factor $(1 + Rtvat)$ at the end of the equation includes all collected VAT \citep{zheng2018tobacco}. Table \ref{tab:1.12} summarizes the allocation-wholesale margin adjustment in 2009, and we can see that high-priced cigarettes have margin rates no lower than low-priced cigarettes.

\begin{table}[htbp]\centering
\caption{Cigarette Allocation Prices and Profit Margins} 
\label{tab:1.12}
\resizebox{\columnwidth}{!}{%
\begin{tabular}{ccccc} \\ \hline \hline
Tier & Allocation Price & \multicolumn{2}{c}{Allocation-Wholesale Margin (\%)} & Wholesale-Retail Margin (\%) \\ \cline{3-4}
                       &                                   & Before May 2009         & After May 2009         &  \\ \hline
I    & {[}10, $\infty$)   & 47  & 31.5                  & 15                 \\
II    & {[}7,10)  & 43   & 25             & 15                 \\
III    & {[}5,7)  &  43  & 25                   &  15                  \\  
III    & {[}3,5)  & 38   & 25             & 10                 \\
IV    & {[}1.65, 3)  & 28   & 20                & 10                 \\
V    & (0, 1.65)     & 18   & 15             & 10    \\
\hline\hline
\multicolumn{5}{p{1\textwidth}}{\footnotesize Notes: Before May 2009, the price range {[}3,5) belonged to Tier III, while the price range {[}5,7) belonged to Tier II. After May 2009, Tier III allocation price falls into the {[}3,7) range. Margin rates are collected from \cite{zheng2018tobacco}.}\\                            
\end{tabular}
}
\end{table}

The calculation procedure for after-tax wholesale price-cost margins involves the following two steps. 
\begin{enumerate}
\item Calculate the allocation price using
\begin{align}
A = \frac{P_w}{(1+a) \times (1+Rtvat)},
\end{align}
where $P_w$ represents the wholesale price in Table \ref{tab:1.2}, $a$ indicates the allocation-wholesale profit margin in Table \ref{tab:1.12}, and $Rtvat$ indicates the 17\% VAT rate.

\item Calculate after-tax wholesale margin using
\begin{align}
\pi_w = A \times a - A \times t_a - t_s,
\end{align}
where $\pi_w$ represents the after-tax wholesale margin, $t_a$ indicates the wholesale ad valorem excise tax rate (i.e. 5\%, as the tax increase in 2015 is passed onto consumers), and $t_s$ indicates the wholesale-specific excise tax listed in Table \ref{tab:1.11} (i.e. \yen 0 in 2011-2014 and \yen 0.1 in 2015-2016).
\end{enumerate}

\section{Decomposition of Price Elasticities}
\label{Decomposition of Elasticities}
Figure \ref{fig:compelas1} summarizes the own-price and cross-price demand elasticities across three models: the standard logit model, the foldable menu model with fixed assortments (where assortments remain unchanged before and after price changes), and the foldable menu model with adjusted assortments (where assortments endogenously adjust according to price changes). The differences between the left and middle panels are due to the variation in parameter estimates between the foldable menu model and the standard logit model, while the differences between the middle and right panels are attributed to the endogenous adjustment of assortments in response to price changes.

\begin{figure}[ht]
    \centering
    \includegraphics[scale=0.5]{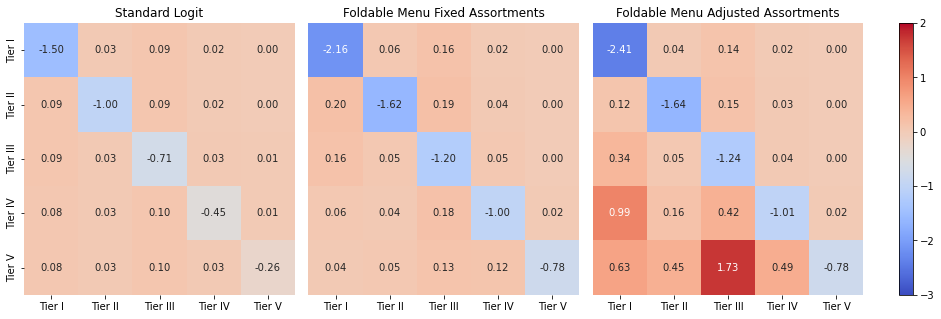}    
    \caption{Decomposition of Price Elasticities. Notes: The columns for Tiers I to V represent cross (and own) elasticities when the price of a specific tier increases by 1\%, while holding the prices of other tiers constant. The left panel presents the values estimated using the standard logit model, the middle panel displays the values estimated using the foldable menu model with fixed assortments, and the right panel shows the values estimated using the foldable menu model with adjusted assortment.}
    \label{fig:compelas1}
\end{figure}

Compared to the foldable menu model with fixed assortments, the standard logit model exhibits lower own- and cross-price elasticities because it underestimates the price sensitivity parameter. However, a price increase in one tier has similar effects on the sales of other tiers in both the left and middle panels. By comparing the middle and right panels, it becomes clear that the firm’s endogenous assortment adjustments are the primary reason why, in the foldable menu model, a price increase in higher-tier cigarettes results in a significantly larger positive effect on the sales of lower-priced cigarettes.

\begin{figure}[ht]
    \centering
    \includegraphics[scale=0.66]{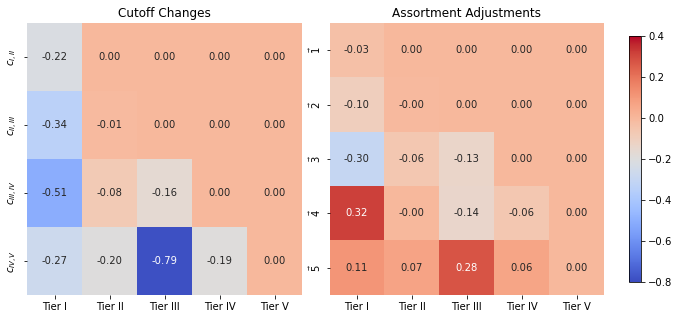}    
    \caption{Effects of Price Increases on Cutoffs and Assortments. Notes: The left panel shows the changes in cutoffs, while the right panel illustrates the corresponding assortment adjustments, both resulting from a 1\% price increase in each cigarette tier. The values for cutoff changes represent the actual changes multiplied by 100, while the values for assortment adjustments reflect the percentage point changes in the proportion of consumers facing each assortment. All values are the mean across all provinces and years.}
    \label{fig:cutoff_assort}
\end{figure}

Figure \ref{fig:cutoff_assort} summarizes the effects of price increases on cutoffs and assortment adjustments. As expected, when the price of Tier $j$ ($j < V$) increases, the cutoffs $c_{j,j+1}, \dots, c_{IV,V}$ decrease, leading to assortment adjustments. Notably, a Tier I price increase has a large positive effect on the ratio of consumers facing $\vec{4}$, and a Tier III price increase has a significant positive effect on the ratio of consumers facing $\vec{5}$. These effects contribute to the larger cross-price elasticities observed for Tier I price increases on Tier IV, and for Tier III price increases on Tier V in the right panel of Figure \ref{fig:compelas1}. Additionally, note that a price increase in Tier $V$ does not affect any cutoffs, so no assortment adjustment occurs, and the estimated elasticities in the last column remain unchanged in the middle and right panels of Figure \ref{fig:compelas1}.

\section{Proof of Lemmas with Competition}
\label{Lemmas with Competition}

\subsection{Lemma 1}

\begin{proof}
\emph{(Lemma 1)}
Under the logit model, we demonstrate that, for any consumer $i$, the inequality $\pi_m \geq E \Pi_{i,\mathcal{J}}$ is equivalent to $E \Pi_{i,\mathcal{J} \cup {m}} \geq E \Pi_{i,\mathcal{J}}$. Let $\mathcal{J}_c$ denote the set of products offered by all competitors.
\begin{align}
E \Pi_{i,\mathcal{J} \cup \{m\}} = &
\sum_{j \in \mathcal{J} \cup \{m\}} \pi_j \frac{\exp(\delta_{ij})}{1+\sum_{j \in \mathcal{J} \cup \{m\} \cup \mathcal{J}_c} \exp(\delta_{ij})} \nonumber \\
= & \sum_{j \in \mathcal{J}} \pi_j \frac{\exp(\delta_{ij})}{1+\sum_{j \in \mathcal{J} \cup \{m\}  \cup \mathcal{J}_c} \exp(\delta_{ij})} + \pi_m \frac{\exp(\delta_{im})}{1+\sum_{j \in \mathcal{J} \cup \{m\}  \cup \mathcal{J}_c} \exp(\delta_{ij})} \nonumber \\
= & E \Pi_{i,\mathcal{J}} \frac{1+\sum_{j \in \mathcal{J}  \cup \mathcal{J}_c} \exp(\delta_{ij})}{1+\sum_{j \in \mathcal{J} \cup \{m\}  \cup \mathcal{J}_c} \exp(\delta_{ij})} + \pi_m \frac{\exp(\delta_{im})}{1+\sum_{j \in \mathcal{J} \cup \{m\}  \cup \mathcal{J}_c} \exp(\delta_{ij})}.
\end{align}
The last equality follows from the definition of $E \Pi_{i,\mathcal{J}}$, where
\[
 E \Pi_{i,\mathcal{J}}  = \sum_{j \in \mathcal{J} } \pi_j \frac{\exp(\delta_{ij})}{1+\sum_{j \in \mathcal{J}  \cup \mathcal{J}_c} \exp(\delta_{ij})}.
\]
Then, we have
\begin{align}
E \Pi_{i,\mathcal{J} \cup \{m\}} - E \Pi_{i,\mathcal{J}} = & E \Pi_{i,\mathcal{J}} \frac{1+\sum_{j \in \mathcal{J}  \cup \mathcal{J}_c} \exp(\delta_{ij})}{1+\sum_{j \in \mathcal{J} \cup \{m\}  \cup \mathcal{J}_c} \exp(\delta_{ij})} + \pi_m \frac{\exp(\delta_{im})}{1+\sum_{j \in \mathcal{J} \cup \{m\}  \cup \mathcal{J}_c} \exp(\delta_{ij})} - E \Pi_{i,\mathcal{J}} \nonumber \\
= & (\pi_m - E \Pi_{i,\mathcal{J}}) \frac{\exp(\delta_{im})}{1+\sum_{j \in \mathcal{J} \cup \{m\}  \cup \mathcal{J}_c} \exp(\delta_{ij})},
\end{align}
so that $\pi_m \geq E \Pi_{i,\mathcal{J}}$ is equivalent to $E \Pi_{i,\mathcal{J} \cup {m}} \geq E \Pi_{i,\mathcal{J}}$. Therefore, it is more profitable to add a new option with a higher profit-margin than the current average profit. Given that we arrange the products in descending order of unit margin, $\pi_1 \geq E \Pi_{i,\mathcal{J}}$ is trivially satisfied, because $\pi_1 \geq \pi_j$ for all $j \in \mathcal{J}$, and $E \Pi_{i,\mathcal{J}}$ represents a weighted average of $\pi_j$ in $\mathcal{J}$. 
\end{proof}

\subsection{Lemma 2}

\begin{proof}
\emph{(Lemma 2)}
First, we show that if $E \Pi_{i,\mathcal{J}\cup\{m\}} \leq E \Pi_{i,\mathcal{J}}$, then $E \Pi_{i, \mathcal{J}\cup\{m,m^{\prime}\}}  \leq E \Pi_{i,\mathcal{J}\cup\{m\}}$ for all $m^{\prime}>m$. That is, if adding $m$ is worse, adding $m^{\prime}$ is worse on top of adding $m$. $E \Pi_{i,\mathcal{J}\cup\{m\}} \leq E \Pi_{i,\mathcal{J}}$ implies that $\pi_{m} \leq  E \Pi_{i,\mathcal{J}}$. We want to show that $\pi_{m^{\prime}} \leq E \Pi_{i,\mathcal{J}\cup\{m\}}$.
\begin{align}
E \Pi_{i,\mathcal{J}\cup\{m\}}  = & \sum_{j\in\mathcal{J}\cup\{m\}}\pi_{j}\frac{\exp(\delta_{ij})}{1+\sum_{j\in\mathcal{J}\cup\{m\}\cup \mathcal{J}_c}\exp(\delta_{ij})}	\nonumber \\
= & \sum_{j\in\mathcal{J}}\pi_{j}\frac{\exp(\delta_{ij})}{1+\sum_{j\in\mathcal{J}\cup\{m\}\cup \mathcal{J}_c}\exp(\delta_{ij})}+\pi_{m}\frac{\exp(\delta_{im})}{1+\sum_{j\in\mathcal{J}\cup\{m\}\cup \mathcal{J}_c}\exp(\delta_{ij})} \nonumber \\
= & E \Pi_{i,\mathcal{J}} \frac{1+\sum_{j \in \mathcal{J}\cup \mathcal{J}_c} \exp(\delta_{ij})}{1+\sum_{j \in \mathcal{J} \cup \{m\}\cup \mathcal{J}_c} \exp(\delta_{ij})} +\pi_{m}\frac{\exp(\delta_{im})}{1+\sum_{j\in\mathcal{J}\cup\{m\}\cup \mathcal{J}_c}\exp(\delta_{ij})} \nonumber \\
\geq & \pi_m \frac{1+\sum_{j\in\mathcal{J}\cup \mathcal{J}_c}\exp(\delta_{ij})}{1+\sum_{j\in\mathcal{J}\cup\{m\}\cup \mathcal{J}_c}\exp(\delta_{ij})}+\pi_{m}\frac{\exp(\delta_{im})}{1+\sum_{j\in\mathcal{J}\cup\{m\}\cup \mathcal{J}_c}\exp(\delta_{ij})}  \nonumber \\
= & \pi_{m}\geq\pi_{m^{\prime}}.
\end{align}
The third equality follows from the definition of $E \Pi_{i,\mathcal{J}}$, and the first inequality follows from $\pi_m \leq E \Pi_{i,\mathcal{J}}$.

Second, we show that if $E \Pi_{i,\mathcal{J}\cup\{m,m^{\prime}\}} \geq E \Pi_{i,\mathcal{J}\cup\{m\}}$, then $E \Pi_{i, \mathcal{J}\cup\{m\}}  \geq E \Pi_{i,\mathcal{J}}$ for all $m < m^{\prime}$. That is, if removing $m^{\prime}$ is worse, removing $m$ is worse on top of removing $m^{\prime}$. $E \Pi_{i,\mathcal{J}\cup\{m,m^{\prime}\}} \geq E \Pi_{i,\mathcal{J}\cup\{m\}}$ implies that $\pi_{m^{\prime}} \geq  E \Pi_{i,\mathcal{J}\cup\{m\}}$. We want to show that $\pi_{m} \geq E \Pi_{i,\mathcal{J}}$. We have
\begin{align}
E \Pi_{i,\mathcal{J}\cup\{m\}}  = & \sum_{j\in\mathcal{J}\cup\{m\}}\pi_{j}\frac{\exp(\delta_{ij})}{1+\sum_{j\in\mathcal{J}\cup\{m\}\cup \mathcal{J}_c}\exp(\delta_{ij})}	\nonumber \\
= & \sum_{j\in\mathcal{J}}\pi_{j}\frac{\exp(\delta_{ij})}{1+\sum_{j\in\mathcal{J}\cup\{m\}\cup \mathcal{J}_c}\exp(\delta_{ij})}+\pi_{m}\frac{\exp(\delta_{im})}{1+\sum_{j\in\mathcal{J}\cup\{m\}\cup \mathcal{J}_c}\exp(\delta_{ij})} \nonumber \\
= & E \Pi_{i,\mathcal{J}} \frac{1+\sum_{j \in \mathcal{J}\cup \mathcal{J}_c} \exp(\delta_{ij})}{1+\sum_{j \in \mathcal{J} \cup \{m\}\cup \mathcal{J}_c} \exp(\delta_{ij})} +\pi_{m}\frac{\exp(\delta_{im})}{1+\sum_{j\in\mathcal{J}\cup\{m\}\cup \mathcal{J}_c}\exp(\delta_{ij})} \nonumber \\
\leq & \pi_{m^{\prime}} \leq \pi_{m}.
\end{align}
The third equality follows from the definition of $E \Pi_{i,\mathcal{J}}$. After moving 
\[
\pi_{m}\frac{\exp(\delta_{im})}{1+\sum_{j\in\mathcal{J}\cup\{m\}\cup \mathcal{J}_c}\exp(\delta_{ij})}
\]
to the RHS of the last inequality, we obtain
\[
E \Pi_{i,\mathcal{J}} \frac{1+\sum_{j \in \mathcal{J}\cup \mathcal{J}_c} \exp(\delta_{ij})}{1+\sum_{j \in \mathcal{J} \cup \{m\}\cup \mathcal{J}_c} \exp(\delta_{ij})} \leq \pi_m \frac{1+\sum_{j \in \mathcal{J}\cup \mathcal{J}_c} \exp(\delta_{ij})}{1+\sum_{j \in \mathcal{J} \cup \{m\}\cup \mathcal{J}_c} \exp(\delta_{ij})},
\]
which implies $\pi_{m} \geq E \Pi_{i,\mathcal{J}}$.
\end{proof}

\subsection{Lemma 3}

\begin{proof}
\emph{(Lemma 3)}
Given Lemma 2, when $\vec{m}$ is optimal among $\{\vec{1},\vec{2},\cdots,\vec{J}\}$, we have $\pi_m \geq E \Pi_{i,\vec{j}}$ for any $j < m$, and $E \Pi_{i,\vec{m}} > \pi_j$ for any $j > m$. Now, we show that $\vec{m}$ also dominates all non-consecutive assortments, such as $\{1,3\}$, $\{1,2,4\}$, etc.

Start with the case $m \geq 3$. First, we compare $\vec{m}$ and its own subsets. Let $\mathbf{k}$ be the set of missing internal products of $\vec{m}$, i.e., $\forall k \in \mathbf{k}, 1<k<m$. Let $\vec{m} \backslash \mathbf{k}$ be the non-consecutive assortment, we have
\begin{align}
& E \Pi_{i,\vec{m}} - E \Pi_{i,\vec{m} \backslash \mathbf{k}} \nonumber \\
= &  \frac{\sum_{j \in \vec{m}} \pi_j \exp(\delta_{ij})}{1 + \sum_{j \in \vec{m} \cup \mathcal{J}_c} \exp(\delta_{ij})} - \frac{\sum_{j \in \vec{m} \backslash \mathbf{k}}\pi_j \exp(\delta_{ij})}{1 + \sum_{j \in \vec{m} \backslash \mathbf{k}\cup \mathcal{J}_c} \exp(\delta_{ij})} \nonumber \\
= & \frac{(\sum_{j \in \vec{m}} \pi_j \exp(\delta_{ij}))(1 + \sum_{j \in \vec{m} \backslash \mathbf{k} \cup \mathcal{J}_c} \exp(\delta_{ij})) - (\sum_{j \in \vec{m} \backslash \mathbf{k}}\pi_j \exp(\delta_{ij}))(1 + \sum_{j \in \vec{m}\cup \mathcal{J}_c} \exp(\delta_{ij}))} {(1 + \sum_{j \in \vec{m}\cup \mathcal{J}_c} \exp(\delta_{ij})) (1 + \sum_{j \in \vec{m} \backslash \mathbf{k} \cup \mathcal{J}_c} \exp(\delta_{ij}))}.
\end{align}
We can simplify the numerator as
\begin{align}
& \sum_{j \in \vec{m}} \pi_j \exp(\delta_{ij}) + [\sum_{j \in \vec{m}} \pi_j \exp(\delta_{ij})][\sum_{j \in \vec{m} \cup \mathcal{J}_c} \exp(\delta_{ij}) - \sum_{j \in  \mathbf{k}} \exp(\delta_{ij})] \nonumber \\
- & \sum_{j \in \vec{m}\backslash \mathbf{k}} \pi_j \exp(\delta_{ij}) - [\sum_{j \in \vec{m}} \pi_j \exp(\delta_{ij}) - \sum_{j \in  \mathbf{k}} \pi_j \exp(\delta_{ij})][\sum_{j \in \vec{m} \cup \mathcal{J}_c} \exp(\delta_{ij})] \nonumber \\
= & \sum_{j \in \mathbf{k}} \pi_j \exp(\delta_{ij}) - [\sum_{j \in \vec{m}} \pi_j \exp(\delta_{ij})][\sum_{j \in  \mathbf{k}} \exp(\delta_{ij})] + [\sum_{j \in  \mathbf{k}} \pi_j \exp(\delta_{ij})][\sum_{j \in \vec{m} \cup \mathcal{J}_c} \exp(\delta_{ij})].
\end{align}
For each $k \in  \mathbf{k}$, we have
\begin{align}
\label{eq:34}
& \pi_k \exp(\delta_{ik}) - [\sum_{j \in \vec{m}} \pi_j \exp(\delta_{ij})]\exp(\delta_{ik}) +  \pi_k \exp(\delta_{ik})[\sum_{j \in \vec{m} \cup \mathcal{J}_c} \exp(\delta_{ij})] \nonumber \\
= & [\pi_k (1+\sum_{j \in \vec{m} \cup \mathcal{J}_c} \exp(\delta_{ij})) - \sum_{j \in \vec{m}} \pi_j \exp(\delta_{ij})]\exp(\delta_{ik})\nonumber \\
= & [\pi_k (1+\sum_{j \in \vec{k} \cup \mathcal{J}_c} \exp(\delta_{ij})) - \sum_{j \in \vec{k}} \pi_j \exp(\delta_{ij}) + \sum_{j=k+1}^m(\pi_k - \pi_j) \exp(\delta_{ij})]\exp(\delta_{ik})\nonumber \\
\geq & [\pi_k (1+\sum_{j \in \vec{k} \cup \mathcal{J}_c} \exp(\delta_{ij})) - \sum_{j \in \vec{k}} \pi_j \exp(\delta_{ij})]\exp(\delta_{ik})\nonumber \\
\geq & [\pi_m (1+\sum_{j \in \vec{k} \cup \mathcal{J}_c} \exp(\delta_{ij})) - \sum_{j \in \vec{k}} \pi_j \exp(\delta_{ij})]\exp(\delta_{ik})\nonumber \\
= & [(\pi_m - E \Pi_{i,\vec{k}})(1+\sum_{j \in \vec{k} \cup \mathcal{J}_c} \exp(\delta_{ij}))]\exp(\delta_{ik})\geq 0.
\end{align}
The second equality follows from  isolating the non-negative term $\sum_{j=k+1}^m(\pi_k - \pi_j) \exp(\delta_{ij})$,  the second inequality follows from $\pi_m \leq \pi_k$ for any $k < m$, the last equality follows from  the definition of $E \Pi_{i,\vec{k}}$, and the last inequality follows from $E \Pi_{i,\vec{k}} \leq \pi_m$ for any $k < m$. So we have $E \Pi_{i,\vec{m}} \geq E \Pi_{i,\vec{m} \backslash \mathbf{k}}$.

Second, for the case $n > m$, let $\mathbf{k}$ be the set of missing internal products of $\vec{n}$, i.e., $\forall k \in \mathbf{k}, 1<k<n$. Let $\vec{n} \backslash \mathbf{k}$ be the non-consecutive assortment, we have
\begin{align}
& E \Pi_{i,\vec{m}} - E \Pi_{i,\vec{n} \backslash \mathbf{k}} \nonumber \\
= &  \frac{\sum_{j \in \vec{m}} \pi_j \exp(\delta_{ij})}{1 + \sum_{j \in \vec{m} \cup \mathcal{J}_c} \exp(\delta_{ij})} - \frac{\sum_{j \in \vec{n} \backslash \mathbf{k}}\pi_j \exp(\delta_{ij})}{1 + \sum_{j \in \vec{n} \backslash \mathbf{k}\cup \mathcal{J}_c} \exp(\delta_{ij})} \nonumber \\
= & \frac{(\sum_{j \in \vec{m}} \pi_j \exp(\delta_{ij}))(1 + \sum_{j \in \vec{n} \backslash \mathbf{k} \cup \mathcal{J}_c} \exp(\delta_{ij})) - (\sum_{j \in \vec{n} \backslash \mathbf{k}}\pi_j \exp(\delta_{ij}))(1 + \sum_{j \in \vec{m}\cup \mathcal{J}_c} \exp(\delta_{ij}))} {(1 + \sum_{j \in \vec{m}\cup \mathcal{J}_c} \exp(\delta_{ij})) (1 + \sum_{j \in \vec{n} \backslash \mathbf{k} \cup \mathcal{J}_c} \exp(\delta_{ij}))}.
\end{align}
We can simplify the numerator as follows:
\begin{align}
\label{eq:36}
& \sum_{j \in \vec{m}} \pi_j \exp(\delta_{ij}) + [\sum_{j \in \vec{m}} \pi_j \exp(\delta_{ij})][\sum_{j \in \vec{n} \cup \mathcal{J}_c} \exp(\delta_{ij}) - \sum_{j \in  \mathbf{k}} \exp(\delta_{ij})] \nonumber \\
- & \sum_{j \in \vec{n}\backslash \mathbf{k}} \pi_j \exp(\delta_{ij}) - [\sum_{j \in \vec{n}} \pi_j \exp(\delta_{ij}) - \sum_{j \in  \mathbf{k}} \pi_j \exp(\delta_{ij})][\sum_{j \in \vec{m} \cup \mathcal{J}_c} \exp(\delta_{ij})] \nonumber \\
= & \sum_{j \in \mathbf{k}} \pi_j \exp(\delta_{ij}) - [\sum_{j \in \vec{m}} \pi_j \exp(\delta_{ij})][\sum_{j \in  \mathbf{k}} \exp(\delta_{ij})] + [\sum_{j \in  \mathbf{k}} \pi_j \exp(\delta_{ij})][\sum_{j \in \vec{m} \cup \mathcal{J}_c} \exp(\delta_{ij})] \nonumber \\
- & \sum_{r = m+1}^n\pi_r \exp(\delta_{ir}) + [\sum_{j \in \vec{m}} \pi_j \exp(\delta_{ij})][\sum_{j \in \vec{n}\cup \mathcal{J}_c} \exp(\delta_{ij})] - [\sum_{j \in \vec{n}} \pi_j \exp(\delta_{ij})][\sum_{j \in \vec{m}\cup \mathcal{J}_c} \exp(\delta_{ij})].
\end{align}
We can further simplify the second line of Equation \eqref{eq:36} as follows:
\begin{align}
& - \sum_{r = m+1}^n\pi_r \exp(\delta_{ir}) + [\sum_{j \in \vec{m}} \pi_j \exp(\delta_{ij})][\sum_{j \in \vec{n}\cup \mathcal{J}_c} \exp(\delta_{ij})] - [\sum_{j \in \vec{n}} \pi_j \exp(\delta_{ij})][\sum_{j \in \vec{m}\cup \mathcal{J}_c} \exp(\delta_{ij})] \nonumber \\
= & - \sum_{r = m+1}^n\pi_r \exp(\delta_{ir}) + [\sum_{j \in \vec{m}} \pi_j \exp(\delta_{ij})][\sum_{j \in \vec{m}\cup \mathcal{J}_c} \exp(\delta_{ij}) + \sum_{r = m+1}^n \exp(\delta_{ir})] \nonumber \\
& - [\sum_{j \in \vec{m}} \pi_j \exp(\delta_{ij}) + \sum_{r = m+1}^n\pi_r \exp(\delta_{ir})][\sum_{j \in \vec{m}\cup \mathcal{J}_c} \exp(\delta_{ij})] \nonumber \\
= & - \sum_{r = m+1}^n\pi_r \exp(\delta_{ir}) + [\sum_{j \in \vec{m}} \pi_j \exp(\delta_{ij})][\sum_{r = m+1}^n \exp(\delta_{ir})] - [\sum_{r = m+1}^n\pi_r \exp(\delta_{ir})][\sum_{j \in \vec{m}\cup \mathcal{J}_c} \exp(\delta_{ij})].
\end{align}
Note for any $r \in [m+1,n]$, we have
\begin{align}
& - \pi_r \exp(\delta_{ir}) + [\sum_{j \in \vec{m}} \pi_j \exp(\delta_{ij})] \exp(\delta_{ir}) - \pi_r \exp(\delta_{ir})[\sum_{j \in \vec{m}\cup \mathcal{J}_c} \exp(\delta_{ij})] \nonumber \\
= & [\sum_{j \in \vec{m}} \pi_j \exp(\delta_{ij}) - \pi_r(1+\sum_{j \in \vec{m}\cup \mathcal{J}_c} \exp(\delta_{ij}))]\exp(\delta_{ir})\nonumber \\
= & (E\Pi_{i,\vec{m}} - \pi_r)(1+\sum_{j \in \vec{m}\cup \mathcal{J}_c} \exp(\delta_{ij}))\exp(\delta_{ir}) > 0.
\end{align}
The last equality follows from the definition of 
$E \Pi_{i,\vec{m}}$, and the last inequality follows from $E \Pi_{i,\vec{m}} > \pi_{r}$ for any $r > m$. 

For any $k > m$, the part $\pi_k \exp(\delta_{ik}) - [\sum_{j \in \vec{m}} \pi_j \exp(\delta_{ij})] \exp(\delta_{ik}) + \pi_k \exp(\delta_{ik})[\sum_{j \in \vec{m}\cup \mathcal{J}_c} \exp(\delta_{ij})]$ will cancel out with one of $-\pi_r \exp(\delta_{ir}) + [\sum_{j \in \vec{m}} \pi_j \exp(\delta_{ij})] \exp(\delta_{ir}) - \pi_r \exp(\delta_{ir})[\sum_{j \in \vec{m}\cup \mathcal{J}_c} \exp(\delta_{ij})]$. For any $k \leq m$, from Equation \eqref{eq:34}, we know
$\pi_k \exp(\delta_{ik}) - [\sum_{j \in \vec{m}} \pi_j \exp(\delta_{ij})] \exp(\delta_{ik}) + \pi_k \exp(\delta_{ik})[\sum_{j \in \vec{m}\cup \mathcal{J}_c} \exp(\delta_{ij})] \geq 0.$ So we have 
$E \Pi_{i,\vec{m}} \geq E \Pi_{i,\vec{n} \backslash \mathbf{k}}$.

Last, for the case $n < m$, let $\mathbf{k}$ be the set of missing internal products of $\vec{n}$, i.e., $\forall k \in \mathbf{k}, 1<k<n$. Let $\vec{n} \backslash \mathbf{k}$ be the non-consecutive assortment, we have
\begin{align}
& E \Pi_{i,\vec{m}} - E \Pi_{i,\vec{n} \backslash \mathbf{k}} \nonumber \\
= &  \frac{\sum_{j \in \vec{m}} \pi_j \exp(\delta_{ij})}{1 + \sum_{j \in \vec{m} \cup \mathcal{J}_c} \exp(\delta_{ij})} - \frac{\sum_{j \in \vec{n} \backslash \mathbf{k}}\pi_j \exp(\delta_{ij})}{1 + \sum_{j \in \vec{n} \backslash \mathbf{k}\cup \mathcal{J}_c} \exp(\delta_{ij})} \nonumber \\
= & \frac{(\sum_{j \in \vec{m}} \pi_j \exp(\delta_{ij}))(1 + \sum_{j \in \vec{n} \backslash \mathbf{k} \cup \mathcal{J}_c} \exp(\delta_{ij})) - (\sum_{j \in \vec{n} \backslash \mathbf{k}}\pi_j \exp(\delta_{ij}))(1 + \sum_{j \in \vec{m}\cup \mathcal{J}_c} \exp(\delta_{ij}))} {(1 + \sum_{j \in \vec{m}\cup \mathcal{J}_c} \exp(\delta_{ij})) (1 + \sum_{j \in \vec{n} \backslash \mathbf{k} \cup \mathcal{J}_c} \exp(\delta_{ij}))}.
\end{align}
We can simplify the numerator as follows:
\begin{align}
\label{eq:40}
& \sum_{j \in \vec{m}} \pi_j \exp(\delta_{ij}) + [\sum_{j \in \vec{m}} \pi_j \exp(\delta_{ij})][\sum_{j \in \vec{n} \cup \mathcal{J}_c} \exp(\delta_{ij}) - \sum_{j \in  \mathbf{k}} \exp(\delta_{ij})] \nonumber \\
- & \sum_{j \in \vec{n}\backslash \mathbf{k}} \pi_j \exp(\delta_{ij}) - [\sum_{j \in \vec{n}} \pi_j \exp(\delta_{ij}) - \sum_{j \in  \mathbf{k}} \pi_j \exp(\delta_{ij})][\sum_{j \in \vec{m} \cup \mathcal{J}_c} \exp(\delta_{ij})] \nonumber \\
= & \sum_{j \in \mathbf{k}} \pi_j \exp(\delta_{ij}) - [\sum_{j \in \vec{m}} \pi_j \exp(\delta_{ij})][\sum_{j \in  \mathbf{k}} \exp(\delta_{ij})] + [\sum_{j \in  \mathbf{k}} \pi_j \exp(\delta_{ij})][\sum_{j \in \vec{m} \cup \mathcal{J}_c} \exp(\delta_{ij})] \nonumber \\
+ & \sum_{r = n+1}^m\pi_r \exp(\delta_{ir}) + [\sum_{j \in \vec{m}} \pi_j \exp(\delta_{ij})][\sum_{j \in \vec{n}\cup \mathcal{J}_c} \exp(\delta_{ij})] - [\sum_{j \in \vec{n}} \pi_j \exp(\delta_{ij})][\sum_{j \in \vec{m}\cup \mathcal{J}_c} \exp(\delta_{ij})]
\end{align}
We can further simplify the second line of Equation \eqref{eq:40} as follows:
\begin{align}
& \sum_{r = n+1}^m\pi_r \exp(\delta_{ir}) + [\sum_{j \in \vec{m}} \pi_j \exp(\delta_{ij})][\sum_{j \in \vec{n}\cup \mathcal{J}_c} \exp(\delta_{ij})] - [\sum_{j \in \vec{n}} \pi_j \exp(\delta_{ij})][\sum_{j \in \vec{m}\cup \mathcal{J}_c} \exp(\delta_{ij})] \nonumber \\
= & \sum_{r = n+1}^m\pi_r \exp(\delta_{ir}) + [\sum_{j \in \vec{n}} \pi_j \exp(\delta_{ij}) + \sum_{r = n+1}^m \pi_r \exp(\delta_{ir})][\sum_{j \in \vec{n}\cup \mathcal{J}_c} \exp(\delta_{ij})] \nonumber \\
- & [\sum_{j \in \vec{n}} \pi_j \exp(\delta_{ij})][\sum_{j \in \vec{n}\cup \mathcal{J}_c} \exp(\delta_{ij}) +  \sum_{r = n+1}^m \exp(\delta_{ir})] \nonumber \\
= & \sum_{r = n+1}^m\pi_r \exp(\delta_{ir}) + [\sum_{r = n+1}^m  \pi_r \exp(\delta_{ir})][\sum_{j \in \vec{n}\cup \mathcal{J}_c} \exp(\delta_{ij})] - [\sum_{j \in \vec{n}} \pi_j \exp(\delta_{ij})][\sum_{r = n+1}^m \exp(\delta_{ir})].
\end{align}
Note for any $r \in [n+1,m]$, we have
\begin{align}
& \pi_r \exp(\delta_{ir}) +  \pi_r \exp(\delta_{ir})[\sum_{j \in \vec{n}\cup \mathcal{J}_c} \exp(\delta_{ij})] - [\sum_{j \in \vec{n}} \pi_j \exp(\delta_{ij})] \exp(\delta_{ir}) \nonumber \\
= & [\pi_r(1+\sum_{j \in \vec{n}\cup \mathcal{J}_c} \exp(\delta_{ij})) - \sum_{j \in \vec{n}} \pi_j \exp(\delta_{ij})]\exp(\delta_{ir})\nonumber \\
= & (\pi_r - E\Pi_{i,\vec{n}})(1+\sum_{j \in \vec{n}\cup \mathcal{J}_c} \exp(\delta_{ij}))\exp(\delta_{ir})\nonumber \\
\geq & (\pi_m - E\Pi_{i,\vec{n}})(1+\sum_{j \in \vec{n}\cup \mathcal{J}_c} \exp(\delta_{ij}))\exp(\delta_{ir}) \geq 0.
\end{align}
The last equality follows from the definition of 
$E \Pi_{i,\vec{n}}$, the first inequality follows from $\pi_r \geq \pi_m$ for any $r \leq m$, and the last inequality follows from $E \Pi_{i,\vec{n}} \leq \pi_m$ for any $n < m$. 

For any $k < m$, from Equation \eqref{eq:34}, we know
$\pi_k \exp(\delta_{ik}) - [\sum_{j \in \vec{m}} \pi_j \exp(\delta_{ij})] \exp(\delta_{ik}) + \pi_k \exp(\delta_{ik})[\sum_{j \in \vec{m}\cup \mathcal{J}_c} \exp(\delta_{ij})] \geq 0.$ So, we have $E \Pi_{i,\vec{m}} \geq E \Pi_{i,\vec{n} \backslash \mathbf{k}}$.
\end{proof} 

\section{Complete Lattice}
\label{Complete Lattice}
Given the foldable menu result, we consider a restricted set of assortments for each firm \( n \) given by: $\{\vec{1}_n, \ldots, \vec{J}_n\}$, where $\vec{j}_n$ represents the top $j_n$ most profitable products for firm $n$. The full set \(\mathcal{L}\) is then the Cartesian product:
\[
\mathcal{L} = \{\vec{1}_1, \ldots, \vec{J}_1\} \times \ldots \times \{\vec{1}_N, \ldots, \vec{J}_N\}.
\]
An element \(\ell \in \mathcal{L}\) is represented as: $\ell = (\ell_1, \ell_2, \ldots, \ell_N)$, where \(\ell_n \in \{\vec{1}_n, \ldots, \vec{J}_n\}\). We define a partial order on \(\mathcal{L}\) as:
\[
\ell \leq \ell' \quad \iff \quad \ell_n \subseteq \ell_n' \,\, \text{for all } n \in \{1, \ldots, N\}.
\]

\textbf{Supremum and Infimum}

The \textbf{supremum} (\textit{join}) of two elements \(\ell, \ell' \in \mathcal{L}\) is the component-wise union:
\[
\sup(\ell, \ell') = (\ell_1 \cup \ell_1', \ell_2 \cup \ell_2', \ldots, \ell_N \cup \ell_N').
\]
Since the sets \(\ell_n\) are nested, the union is always one of the elements in the set.

The \textbf{infimum} (\textit{meet}) of two elements \(\ell, \ell' \in \mathcal{L}\) is the component-wise intersection:
\[
\inf(\ell, \ell') = (\ell_1 \cap \ell_1', \ell_2 \cap \ell_2', \ldots, \ell_N \cap \ell_N').
\]
Similarly, the intersection is always one of the elements in the set.

Since every subset of \(\mathcal{L}\) has both a supremum and an infimum under the partial order defined by subset inclusion, \(\mathcal{L}\) forms a \textbf{complete lattice}.

\section{Simulations: Foldable Menu in Random Coefficient Models} \label{SimuRandCoeff}

Using the application for China's tobacco industry, we simulate 20 points of sale for each province-year combination based on income distributions. We then assign the corresponding estimated parameters \(\gamma\) and \(\theta\) to each point of sale. This results in a total of \(31 \times 6 \times 20 = 3720\) points of sale per simulation.

We use the variable \(n\) to denote the number of random draws for each point of sale and \(a\) to measure the magnitude of random taste variations relative to the constant taste shared by all consumers. For example, at a point of sale with mean income \(Inc_i\), the mean price sensitivity is given by \(\alpha_i = \theta / Inc_i\). The standard deviation of random price sensitivity is calculated as \(\sigma_{\alpha_i} = a \times \theta / Inc_i\), and consumer \(d\)'s random price sensitivity is expressed as \(\sigma_{\alpha_i} v_d\), where \(v_d \sim N(0,1)\). Similarly, the standard deviation of random utility is derived as \(\sigma_{\gamma_{jlt}} = a \times \gamma_{jlt}\).

We use \(n \in \{1,000, 2,000\}\), values that fall within the range commonly applied in studies comparing BLP estimation procedures. For instance, \cite{dube2012improving} employs \(n = 1,000\) or \(n = 3,000\), while \cite{pal2023comparing} uses \(n = 1,331\). For the parameter \(a\), we consider values \(a \in \{0.5, 1, 2, 5\}\), which span the range of estimates reported in prior studies. For example, in \cite{berry1995automobile}, the estimated parameters are \(\beta = \{-7.061, 2.883, 1.521, -0.122, 3.460\}\) and \(\sigma = \{3.612, 4.628, 1.818, 1.050, 2.056\}\) for the variables Constant, HP/Weight, Air, MP\$, and Size, respectively. Notably, the \(\beta\) coefficient for MP\$ is statistically insignificant. For the remaining four variables, the ratio \(|\sigma / \beta|\) ranges from 0.51 to 1.61, with a mean of 0.98. Similarly, in \cite{nevo2001measuring}, the ratio \(|\sigma / \beta|\) ranges from 0.09 to 1.99, with a mean of 0.52.

We present the performance of foldable menu assortments in Figure \ref{fig:perforfm}. For the chosen values of \(n\) and \(a\), the foldable menu's highest profit closely approximates the highest profit achievable from all assortments, as evidenced by the very small percentage loss in optimal profit.

\begin{figure}[ht]
   \centering
    \includegraphics[scale=0.3]{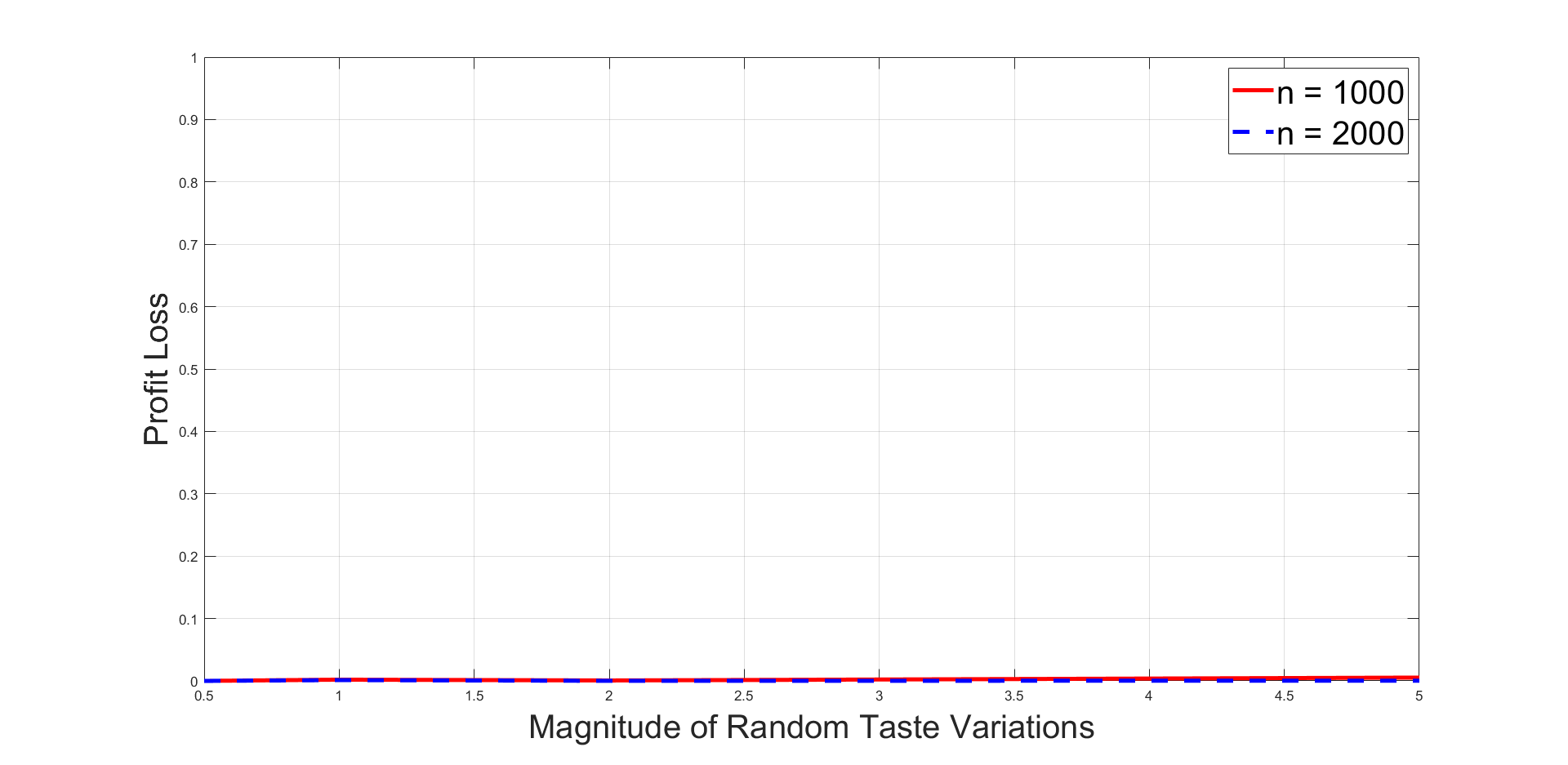} 
\caption{Performance of Foldable Menu Assortments. Notes: The \( x \)-values represent the magnitudes of random taste variations, taking values from \( \{0.5, 1, 2, 5\} \). The \( y \)-values represent the percentage loss in optimal profit, calculated as 1 minus the mean ratio of the foldable menu's highest profit to the highest profit achievable from all assortments.}
\label{fig:perforfm}
\end{figure}

\end{appendices}

\newpage

\bibliographystyle{chicago}
\bibliography{refmain}

\end{document}